\documentclass[11pt]{article}
\pdfoutput=1
\usepackage{soul}
\usepackage{jheppub}
\usepackage{amsfonts}
\usepackage{amsmath}
\usepackage{amsthm}
\allowdisplaybreaks[4]         
\usepackage{amssymb}   
\usepackage{euscript}   
\usepackage{cleveref}    
\usepackage[dvipsnames]{xcolor}          
\usepackage{tensor}     
\usepackage{graphicx}
\usepackage{graphicx}
\usepackage{caption}
\usepackage{subcaption}   
\usepackage{hyperref}
\hypersetup{
	colorlinks =NavyBlue,
	linkcolor =NavyBlue,
	pdftitle={},
	citecolor=MidnightBlue,
	filecolor=NavyBlue,
	urlcolor=NavyBlue,
}

\newcommand{\bea}{\begin{eqnarray}}
\newcommand{\eea}{\end{eqnarray}}
\newcommand{\ba}{\begin{eqnarray}}
\usepackage{braket}
\newcommand{\ea}{\end{eqnarray}}

\newcommand{\beq}{\begin{equation}}
\newcommand{\eeq}{\end{equation} }
\newcommand{\beqa}{\begin{eqnarray}}

\newcommand{\eeqa}{\end{eqnarray}}
\newcommand{\beqar}{\begin{eqnarray*}}
\newcommand{\eeqar}{\end{eqnarray*}}

\newcommand{\be}{\begin{equation}}
\newcommand{\ee}{\end{equation}}
\newcommand{\diff}{\mathrm{d}}

\newcommand{\Rho}{\mathrm{P}}

\newtheorem{theorem}{Theorem}

\newtheorem{defi}{Definition}

\newtheorem{proposition}{Proposition}
\newtheorem{example}{Example}

\newcommand{\I}{\mathcal{I}}

\newcommand{\X}{\mathcal{X}}

\newcommand{\R}{\mathcal{R}}
\newcommand{\Z}{\mathcal{Z}}
\newcommand{\cS}{\mathcal{S}}

\newcommand{\sZ}{\mathsf{Z}}
\newcommand{\sW}{\mathsf{W}}
\newcommand{\sX}{\mathsf{X}}
\newcommand{\sY}{\mathsf{Y}}

\newcommand{\cL}{\mathcal{L}}

\definecolor{shadecolor}{rgb}{.25,.25,.25}

\title{ \boldmath On the classification of Generalized Quasitopological Gravities
}

\author[a,b]{Javier Moreno,}
\author[c]{\'Angel J. Murcia}

\affiliation[a]{Department of Physics and Haifa Research Center for Theoretical Physics and Astrophysics\\
University of Haifa, Haifa 31905, Israel \vspace{0.1cm}}
\affiliation[b]{Department of Physics, Technion, Israel Institute of Technology,
Haifa, 32000, Israel \vspace{0.1cm}}
\affiliation[c]{Istituto Nazionale di Fisica Nucleare, Sezione di Padova, Via Marzolo 8, 35131 Padova, Italy \vspace{0.1cm}}

\emailAdd{jmoreno@campus.haifa.ac.il}
\emailAdd{angel.murcia@pd.infn.it}

\date{\today}
\abstract{Generalized Quasitopological Gravities (GQTGs) are higher-order extensions of Einstein gravity in $D$ dimensions satisfying a number of interesting properties, such as possessing second-order linearized equations of motion on top of maximally symmetric backgrounds, admitting non-hairy generalizations of the Schwarzschild-Tangherlini black hole which are characterized by a single metric function or forming a perturbative spanning set of the space of effective theories of gravity. In this work, we classify all inequivalent GQTGs at all curvature orders $n$ and spacetime dimension $D \geq 4$. This is achieved after the explicit construction of a dictionary that allows the uplift of expressions evaluated on a single-function static and spherically symmetric ansatz into fully covariant ones. On the one hand, applying such prescription for $D \geq 5$, we find the explicit covariant form of the unique inequivalent Quasitopological Gravity that exists at each $n$ and, for the first time, the covariant expressions of the $n-2$ inequivalent proper GQTGs existing at every curvature order $n$. On the other hand, for $D=4$, we are able to provide the first rigorous proof of the fact that there is one and only one (proper) inequivalent GQTG at each curvature order $n$, deriving along the way a simple expression for such four-dimensional representative at every order $n$.

 }

\begin{document} 
\maketitle
\flushbottom


\section{Introduction}
\label{sec:Introduction}

Higher-order gravities --- or equivalently, higher-curvature gravities or higher-derivative gravities --- are extensions of General Relativity in which the Einstein-Hilbert action is supplemented by terms of higher-order in the curvature of spacetime.  They are natural modifications of Einstein gravity with which to capture and parametrize the effects of an underlying UV-complete theory (Quantum Gravity). As a matter of fact, they arise as stringy corrections in the low-energy effective actions of the different versions of String Theory \cite{Callan:1985ia,Zwiebach:1985uq,Grisaru:1986vi,Gross:1986iv,Gross:1986mw,Bergshoeff:1989de}. Also, they appear naturally  from an Effective Field Theory (EFT) perspective, in which one considers all possible terms which are compatible with the symmetries of theory \cite{Weinberg:1995mt}. When dealing with purely gravitational actions, such EFT approach corresponds precisely to the introduction of higher-curvature terms, which preserve --- of course --- diffeomorphism invariance.

Thus, in recent years, higher-order gravities have acquired a genuine interest in the literature, becoming a topic of research in their own right. The reasons for this ever-growing attention have been manifold. Firstly, the introduction of higher-derivative terms generically gives rise to renormalizable actions \cite{Stelle:1976gc, Stelle:1977ry}, so that it is interesting to examine how higher-derivative terms correct Einstein-gravity solutions in situations of sufficiently large curvature, such as black holes or the early Universe \cite{Starobinsky:1980te,Wheeler:1985nh,Boulware:1985wk,Wiltshire:1988uq,Myers:1988ze,Myers:2010ru,Lu:2015cqa,Bueno:2016lrh,Bueno:2017sui}. Secondly, higher-order gravities have turned out to be extremely useful in the holographic context, via the AdS/CFT correspondence \cite{Maldacena:1997re,Gubser:1998bc,Witten:1998qj}. Apart from capturing finite $N$ and finite-coupling effects within the canonical correspondence between IIB String Theory and $\mathcal{N}=4$ Super-Yang-Mills theory \cite{Gubser:1998nz,Buchel:2004di,Myers:2008yi}, they allow one to explore CFTs whose correlators adopt the most generic form permitted by conformal symmetry \cite{Hofman:2008ar,deBoer:2009pn,Camanho:2009vw,Buchel:2009sk,deBoer:2009gx,Myers:2010jv,Bueno:2018xqc,Cano:2022ord} or to discover universal features of generic CFTs \cite{Brigante:2007nu,Myers:2008yi,Cai:2008ph,Myers:2010tj,Myers:2010xs,Perlmutter:2013gua,Mezei:2014zla,Bueno:2015rda,Bueno:2015xda,Chu:2016tps,Dey:2016pei,Bueno:2018yzo,Li:2018drw,Bueno:2020uxs,Bueno:2020odt,Anastasiou:2021swo,Anastasiou:2021jcv,Bueno:2022jbl,Murcia:2023zok}. And, in third place, given the increasing accuracy of current gravitational-wave observations from neutron-star and black-hole binaries mergers \cite{LIGOScientific:2018mvr,LIGOScientific:2020ibl,LIGOScientific:2021djp} and the potential measurements that could be carried out from the imaging of black-hole shadows \cite{EventHorizonTelescope:2019dse,EventHorizonTelescope:2022xnr}, higher-order gravities may be extremely helpful in the search and study of deviations from Einstein gravity in astrophysical observables \cite{Cardoso:2009pk,Blazquez-Salcedo:2017txk,Berti:2018cxi,Cardoso:2018ptl,Cano:2021myl,Silva:2022srr}.

Well-known examples of higher-curvature gravities that have been exhaustively studied in the literature are Lovelock gravities \cite{lovelock1970divergence,Lovelock:1971yv} and $f(R)$ gravities \cite{Buchdahl:1970ynr,Sotiriou:2008rp,DeFelice:2010aj}. On the one hand, Lovelock theories correspond to the most general diffeomorphism-invariant theories which possess gravitational equations of motion of second derivative order. Nevertheless, they reduce to Einstein gravity in four dimensions, up to topological terms. On the other hand, $f(R)$ gravities are defined as those theories whose Lagrangian is an arbitrary function of the Ricci scalar. In this case, they turn out to be classically equivalent to Brans-Dicke theories \cite{Teyssandier:1983zz,Barrow:1988xh} and, furthermore, every asymptotically flat vacuum solution of GR is also a solution of $f(R)$ gravities\footnote{Whenever it is assumed that $f(R)$ admits a polynomial expansion around $R=0$ and satisfies that $f(0)=0$.}, so they do not actually introduce new purely gravitational phenomena. Therefore, it is necessary to consider more general higher-curvature theories in order to correct purely gravitational four-dimensional GR solutions. However, if one deals with a generic higher-order gravity, the corresponding equations of motion will be of fourth order in derivatives\footnote{If covariant derivatives of the curvature appear, such order is further increased.} and its resolution becomes an extremely daunting (if not impossible) task. 

Such discouraging landscape could be partly blurred away if it was possible to find higher-derivative gravities with second-order equations of motion, at least, in highly symmetric configurations, such as static and spherically symmetric (SSS) ones. Indeed, they correspond to realistic situations but, still, possess enough symmetry so as to simplify the subsequent computations. This eclectic posture has turned out to be extraordinarily successful. First, the class of Quasitopological Gravities (QTGs) was identified \cite{Oliva:2010eb,Myers:2010ru}. They are characterized by admitting SSS solutions which are fully specified by a single function --- like the Schwarzschild-Tangherlini GR solution --- whose equation of motion is algebraic. Apart from Einstein gravity and Lovelock theories, further explicit examples of QTGs have been constructed to cubic order in the curvature \cite{Oliva:2010eb,Myers:2010ru}, quartic \cite{Dehghani:2011vu}, quintic \cite{Cisterna:2017umf} and to arbitrary order \cite{Bueno:2019ycr}. However, no QTG exists in four dimensions. This problem can be circumvented if one only demands the existence of SSS solutions characterized by a single function whose equation is at most of second order in derivatives. This defines the class of Generalized Quasitopological Gravities (GQTGs) \cite{Bueno:2016xff,Hennigar:2017ego,Bueno:2017sui,Ahmed:2017jod}, which encompasses all QTGs. Interestingly enough, there are non-trivial instances of these theories in four dimensions, such as the celebrated Einsteinian Cubic Gravity \cite{Bueno:2016xff}. Also, it has been possible to construct examples of proper GQTGs (i.e., not belonging to the QTG class) in every dimension $D \geq 5$ and to all orders in curvature \cite{Bueno:2019ycr}. Additionally, GQTGs possess linearized equations of motion around maximally symmetric backgrounds \cite{Bueno:2016xff,Hennigar:2016gkm,Bueno:2016lrh,Hennigar:2017ego,Bueno:2017sui,Ahmed:2017jod}, allow the exact computation of black-hole thermodynamics \cite{Oliva:2010eb,Myers:2010ru,Hennigar:2017ego,Bueno:2017qce,Cisterna:2017umf,Bueno:2017sui,Bueno:2017qce} and form a basis for the whole space of higher-curvature gravities if field redefinitions of the metric are considered \cite{Bueno:2019ltp}.

Therefore, a complete characterization of GQTGs at all curvature orders and for every dimension\footnote{In $D=3$, all GQTG theories except for Einstein gravity vanish when evaluated on the single-function SSS ansatz \cite{Bueno:2022lhf}.} $D \geq 4$ is clearly of interest.\footnote{We shall assume the absence of covariant derivatives of the curvature in the Lagrangian.} Some steps in this direction has been carried out in the literature, such as the derivation of an explicit --- although somewhat involved --- formula for a QTG at every curvature order for $D \geq 5$ \cite{Bueno:2019ycr}, the construction of a proper GQTG --- through a considerably difficult process as well ---  at every order and dimension $D \geq 4$ \cite{Bueno:2019ycr} and the characterization of all GQTGs when evaluated on an SSS ansatz specified by a single function \cite{Bueno:2022res}. In this Ref., it was explicitly proven that, up to the addition of densities that vanish when evaluated on the single-function SSS ansatz (trivial GQTGs),  there exist $n-1$ GQTGs at every curvature order $n$ in $D \geq 5$, only one of them being of the quasitopological type. However, the finding of the covariant form of all such $n-1$ GQTGs (up to trivial GQTGs) has remained elusive. Furthermore, despite the strong evidence that there exists a unique GQTG at each curvature order for $D=4$ (again, up to trivial GQTGs), no rigorous proof of this has been provided yet. 


It is the purpose of this manuscript to fill the previous gaps and conclude the classification of all GQTGs for every curvature order $n$ and every $D \geq 4$. Among other reasons, this is motivated by the fact that this would allow one understand the nature of GQTGs away from the SSS regime, for instance considering cosmological solutions or rotating black holes, as in \cite{Arciniega:2018fxj,Arciniega:2018tnn,Cisterna:2018tgx,Cano:2019ozf,Adair:2020vso,Cano:2020oaa,Cano:2023dyg}. After the development of a dictionary that allows one to pass from expressions evaluated on the single-function SSS ansatz to covariant (or off-shell) expressions made of combinations of specific curvature invariants, we are able to construct (up to the addition of trivial GQTGs) a simple expression for the explicit covariant form of the unique QTG existing at each $n$ and $D \geq 5$ --- see Theorem \ref{theo:qtgs} ---, derive for the first time the off-shell expression of all $n-2$ proper GQTGs at every $n$ and $D \geq 5$ --- see Theorem \ref{theo:gqtgs} --- and rigorously prove that there exists a unique (proper) GQTG in $D \geq 4$ at each curvature order, also providing its explicit covariant form --- see Theorem \ref{theo:d4gqtgs}. 


The article is organized as follows. In Section \ref{sec:GQTGs} we review the definition of GQTGs and mention their most salient features. Next, in Section \ref{sec:offon} we present a procedure by which to uplift expressions evaluated on a single-function SSS ansatz to fully covariant combinations of curvature invariants in $D \geq 5$, see Proposition \ref{prop:wzr} and the dictionary given in Eqs. \eqref{eq:dic1} to \eqref{eq:dic5}. Afterwards, using such dictionary, in Section \ref{sec:all} all QTGs and proper GQTGs (up to trivial GQTGs) are presented at each curvature order and $D \geq 5$, analyzing the number of curvature invariants needed to construct them. Then, in Section \ref{sec:gqtgs4} we adapt the aforementioned dictionary to $D=4$, prove that there exists a unique GQTG at each $n$ in four dimensions up to the addition of trivial GQTGs and construct an explicit example of GQTG at each curvature order. Finally, we provide some conclusions and future directions, including some supplementary material in the appendix.

\subsubsection*{Note on conventions}

All along the document we will use the mostly plus signature for the metric $g_{ab}$ and the conventions from Wald's book \cite{Wald:1984rg} for the curvature. In particular, the Riemann curvature tensor $R_{abc}{}^d$ associated to the Levi-Civita connection of $g_{ab}$ is given by:
\begin{equation}
R_{abc}{}^{d}=-2\partial_{[a} \Gamma^{d}{}_{b]c}+ 2\Gamma^{e}{}_{c[a} \Gamma^{d}{}_{b]e}\,,
\end{equation}
where $\Gamma^{c}{}_{ab}$ stands for the Christoffel symbols and where $[ab]$ indicates as usual the antisymmetric part --- for a two-tensor $A_{ab}$, $A_{[ab]}=1/2(A_{ab}-A_{ba})$. Our definitions for the Ricci tensor $R_{ab}$ and the Ricci scalar $R$ are:
\begin{equation}
R_{ab}=R_{acb}{}^c\, , \qquad R=R_{a}{}^{a}\,. 
\end{equation}
We will also be using the Weyl tensor $W_{abc}{}^d$ and the traceless Ricci tensor $Z_{ab}$. They are defined as follows for a $D$-dimensional manifold:
\begin{align}
W_{abcd}&=R_{abcd}-\frac{2}{D-2}(g_{a[c} R_{d]b}-g_{b[c} R_{d]a})+\frac{2}{(D-1)(D-2)} R g_{a[c}g_{d]b}\,,\\
Z_{ab}&=R_{ab}-\frac{1}{D}g_{ab} R\,.
\end{align}


\section{Generalized Quasitopological Gravities}\label{sec:GQTGs}

Let $\cL(g^{ab},R_{a b c d})$ be a generic diffeomorphism-invariant theory of gravity constructed from\footnote{For the sake of simplicity, we will assume that no covariant derivatives of the curvature appear on the action. We will also assume the absence of parity-breaking terms.} contractions of the Riemann curvature tensor $R_{a b c d}$ associated to a metric $g_{ab}$. Assuming that $\cL$ admits an effective expansion in increasing powers of the curvature (or equivalently, derivatives of the metric), we may write
\begin{equation}
\cL(g^{ab},R_{a b c d})=-2\Lambda+ R+\sum_{n=2}^\infty \sum_{j=1}^{k_{n}} \ell^{2(n-1)}\beta_{n,j}  \mathcal{R}_{(n,j)}\,,
\label{eq:effgrav}
\end{equation}
where $\Lambda$ is the cosmological constant, $\ell$ is a characteristic length scale from which on the effects of higher-derivative terms have to be taken into account, $\beta_{n,j}$ are dimensionless couplings and $\mathcal{R}_{(n,j)}$ stands for the different contractions of curvature tensors at order $2n$ in derivatives that one may construct, each of them labeled by $j$.  

Now, consider a general static and spherically symmetric (SSS) $D$-dimensional configuration. In an appropriate coordinate system, it can be expressed in terms of two unknown functions $N(r)$ and $f(r)$ as follows:
\begin{equation}
\label{eq:sssans}
\diff s_{N,f}^2=-N^2(r) f(r) \diff t^2+ \frac{\diff r^2}{f(r)}+r^2 \diff \Omega^2_{D-2}\,,
\end{equation}
where $\diff \Omega^2_{D-2}$ stands for the round metric of the $(D-2)$-dimensional sphere. Also, let $L_{N,f}= r^{D-2} N(r) \cL \vert_{N,f}$ be the effective Lagrangian evaluated on the SSS ansatz \eqref{eq:sssans} and $L_{f}= L_{1,f}$.
\begin{defi}
A theory $\cL( g^{ab},R_{a b c d})$ is said to be a Generalized Quasitopological Gravity (in short, GQTG) if and only if
\begin{equation}
\frac{\partial L_f}{\partial f}-\frac{\diff}{\diff r} \frac{\partial L_f}{\partial f'}+\frac{\diff^2}{\diff r^2} \frac{\partial L_f}{\partial f''}-\dots=0\,.
\label{eq:gqtgdef}
\end{equation}
\end{defi}
Such definition was put forward first in \cite{Hennigar:2017ego}. By now, the consequences arising from demanding \eqref{eq:gqtgdef} to hold have been extensively explored in the literature (see e.g. the vast number of Refs. presented in the enumeration below). Here we quote the most relevant ones: 
\begin{enumerate}
\item The equations of motion of GQTGs are second order in derivatives when linearized around maximally symmetric backgrounds. As a consequence, they only propagate a massless and traceless graviton on such backgrounds \cite{Bueno:2016xff,Hennigar:2016gkm,Bueno:2016lrh,Hennigar:2017ego,Bueno:2017sui,Ahmed:2017jod,Bueno:2017qce}.

\item They admit (asymptotically flat, de Sitter or Anti-de Sitter) SSS solutions characterized by $N(r)=1$ and having (at most) second-order equation in derivatives for the ansatz function $f(r)$. These solutions are interpreted as natural generalizations of the (asymptotically flat, de Sitter or Anti-de Sitter) Schwarzschild-Tangherlini solution of Einstein gravity, since they turn out to be completely characterized by its Arnowitt-Deser-Misner (ADM) mass \cite{Arnowitt:1960es,Arnowitt:1960zzc,Arnowitt:1961zz}. The equation for $f(r)$ is  obtained by varying $L_{N,f}$ with respect to $N(r)$ and afterwards imposing $N(r)=1$:
\begin{equation}\label{eq:eqN}
\mathcal{E}^{\text{SSS}}=\left. \frac{\partial L_{N,f}}{\partial N}\right \vert_{N=1}-\left.\frac{\diff}{\diff r} \frac{\partial L_{N,f}}{\partial N'}\right \vert_{N=1}+\left.\frac{\diff^2}{\diff r^2} \frac{\partial L_{N,f}}{\partial N''}\right \vert_{N=1}-\dots =0\,.
\end{equation} 
According to the order (in derivatives) of $\mathcal{E}^{\text{SSS}}$, GQTGs may be divided into the two following subfamilies:
\begin{itemize}
\item Those for which $\mathcal{E}^{\text{SSS}}$ is an algebraic equation for $f(r)$. These theories are called \emph{Quasitopological Gravities} (QTGs) \cite{Oliva:2010eb,Myers:2010ru,Dehghani:2011vu,Ahmed:2017jod,Cisterna:2017umf}. They exist for $D \geq 5$.
\item Those for which $\mathcal{E}^{\text{SSS}}$ is of second order in derivatives. The corresponding theories are called \emph{proper GQTGs} (or \emph{genuine GQTGs}) or, whenever no confusion may arise, just \emph{GQTGs} \cite{Bueno:2016xff,Hennigar:2017ego,Bueno:2017sui,Bueno:2022res}. They have been shown to exist for $D \geq 4$.
\end{itemize}
Additionally, we may identify the class of \emph{trivial GQTGs}, defined as those vanishing identically when evaluated on the SSS ansatz \eqref{eq:sssans} with $N(r)=1$.
\item GQTGs admit solutions with $(D-2)$-dimensional hyperbolic and planar sections too (i.e. solutions of the form \eqref{eq:sssans} but replacing $\diff \Omega^2_{D-2}$ by the corresponding hyperbolic or planar metrics), with $N(r)=1$ and second-order equation (at most) for $f(r)$.

\item Black-hole thermodynamics can be computed analytically \cite{Oliva:2010eb,Myers:2010ru,Hennigar:2017ego,Bueno:2017qce,Cisterna:2017umf,Bueno:2017sui,Bueno:2019ycr,Cano:2019ozf,Frassino:2020zuv,Khodabakhshi:2020hny,KordZangeneh:2020qeg,Bueno:2022res}. This is also the case for the extended thermodynamics approach, in which the the cosmological constant is interpreted as the pressure of the black hole \cite{Mir:2019ecg,Mir:2019rik}.  


\item Any gravitational effective action  of the form \eqref{eq:effgrav} can be mapped, perturbatively order by order, to a certain GQTG \cite{Bueno:2019ltp}. Thus, GQTGs form a spanning set of the space of effective field theories of gravity which is specially suited for the study of SSS configurations.

\item There are particular subsets of GQTGs that additionally allow for reduction-of-order in the equations of motions on other gravitational configurations, like Taub-NUT/Bolt metrics \cite{Bueno:2018uoy}, wormhole geometries \cite{Mehdizadeh:2019qvc,Mustafa:2020qjo} or cosmological backgrounds \cite{Feng:2017tev,Arciniega:2018fxj,Cisterna:2018tgx,Arciniega:2018tnn,Quiros:2020uhr,Marciu:2020ski,Quiros:2020eim,Edelstein:2020nhg,Edelstein:2020lgv,BeltranJimenez:2020lee,Cano:2020oaa}.
\end{enumerate}

Apart from the aforementioned properties, there are two additional aspects that are worth to mention. On the one hand, it turns out that GQTGs are very useful in the holographic context as well, since they allow one to explore CFTs whose correlators take the most generic form allowed by conformal symmetry \cite{Myers:2010jv,Bueno:2018xqc,Li:2018drw,Cano:2022ord}, identify novel universal relations that hold for arbitrary CFTs \cite{Bueno:2018yzo,Bueno:2020odt,Bueno:2022jbl}, examine aspects of holographic entanglement entropies \cite{Dey:2016pei,Bueno:2020uxs,Caceres:2020jrf,Anastasiou:2022pzm} or study generic features of holographic transport and superconductivity \cite{Mir:2019ecg,Mir:2019rik,Edelstein:2022xlb,Murcia:2023zok}. On the other hand, it is possible to extend the definition for GQTGs provided in \eqref{eq:gqtgdef} to allow for the inclusion of non-minimally coupled matter, such as a $\mathrm{U}(1)$ vector field \cite{Cano:2020ezi,Cano:2020qhy,Bueno:2021krl,Cano:2022ord,Bueno:2022ewf}. 

\begin{defi}
\label{def:gqtgeq}
Let $\cL^{(1)}$ and $\cL^{(2)}$ be two non-trivial GQTGs. Assume the SSS ansatz \eqref{eq:sssans} with $N(r)=1$ and let $\mathcal{E}^{\text{SSS}}_1$ and $\mathcal{E}^{\text{SSS}}_2$ be the corresponding equations of motion for $f(r)$. Then, $\cL^{(1)}$ and $\cL^{(2)}$ are said to be equivalent if $\mathcal{E}^{\text{SSS}}_1$ and $\mathcal{E}^{\text{SSS}}_2$ are linearly dependent. Otherwise, we say they are inequivalent. 
\end{defi}
Definition \eqref{def:gqtgeq} naturally introduces classes of equivalence of GQTGs, two GQTGs belonging to the same class if and only if they are equivalent. When working with SSS configurations \eqref{eq:sssans}, it suffices to work with a representative of each equivalence class. By convention, we will also assume that trivial GQTGs are all equivalent to $0$ in the sense we have just introduced.

For $D \geq 5$, it was proven in \cite{Bueno:2022res} that there exist $n-1$ inequivalent GQTGs constructed solely from contractions of $n$ curvature tensors, one of them being a QTG and the remaining $n-2$ ones being proper (or genuine) GQTGs. At $D=4$ instead, there seems to be a single proper GQTG at each order, although no rigorous proof has been offered up to date, to the best of our knowledge. For $D=3$, no non-trivial GQTGs exist (apart from the Einstein-Hilbert term, of course).



Taking into account the previous considerations and that the GQTG condition expressed in \eqref{eq:gqtgdef} is linear, the most general combination of all inequivalent GQTGs can be expressed as follows:
\begin{align}
\cL_{\mathrm{GQTG}}&=-2 \Lambda +R+ \sum_{n=2}^\infty  \ell^{2(n-1)} \alpha_{n}   \mathcal{Z}_{(n)}+\sum_{n=3}^\infty \sum_{j=2}^{n-1} \ell^{2(n-1)}\beta_{n,j}  \mathcal{S}_{(n,j)}\,, \quad D \geq 5\,, \\ 
\cL_{\mathrm{GQTG}}&=-2 \Lambda +R+\sum_{n=3}^\infty \ell^{2(n-1)}\beta_{n}  \mathcal{S}_{(n)}\,, \quad D=4\,,
\end{align}
where $\mathcal{Z}_{(n)}$ and $\mathcal{S}_{(n,j)}$ stand for particular representatives of QTGs and proper GQTGs,\footnote{Observe that the sum on proper GQTGs starts from $n=3$, since no proper quadratic GQTG exists.} respectively. Consequently, finding the expression of all inequivalent GQTGs requires the knowledge of a representative of every equivalence class of GQTGs at every order in the curvature.

GQTGs comprise the well-known Lovelock gravities\footnote{In particular, whenever a Lovelock gravity of order $n$ in spacetime dimension $D \geq 5$ does not vanish identically nor it is topological, it belongs to the equivalence class of QTGs of order $n$.} \cite{lovelock1970divergence,Lovelock:1971yv,Wheeler:1985nh,Boulware:1985wk,Cai:2001dz,Padmanabhan:2013xyr} (which contains Einstein gravity), but there exist of course explicit instances of GQTGs which are not Lovelock theories. Let us show an example of a (non-Lovelock) QTG and of a proper GQTG.
\begin{example}
Let us consider the following density\footnote{In this section, we will use some appropriate superindices on densities to make reference to the initials of the authors that introduced them.} in $D\geq 5$ \cite{Myers:2010ru}:
\begin{align}
\mathcal{Z}_{(3)}^{\mathrm{MR}}=&{{{R_a}^b}_c{}^d} {{{R_b}^e}_d}^f {{{R_e}^a}_f}^c
               + \frac{1}{(2D - 3)(D - 4)} \left[ \frac{3(3D - 8)}{8} R_{a b c d} R^{a b c d} R \nonumber \right. \\ \label{eq:Zmyers}
              & \left. - \frac{3(3D-4)}{2} {R_a}^c {R_c}^a R  - 3(D-2) R_{a c b d} {R^{a c b}}_e R^{d e} + 3D R_{a c b d} R^{a b} R^{c d} \right. \\ \nonumber
              & \left.
                + 6(D-2) {R_a}^c {R_c}^b {R_b}^a  + \frac{3D}{8} R^3 \right]\, .
\end{align}
It can be checked that the theory defined by \eqref{eq:Zmyers} satisfies the condition \eqref{eq:gqtgdef} when evaluated on the ansatz \eqref{eq:sssans} with $N(r)=1$, so it belongs to the GQTG family. Further inspection reveals that the subsequent equation of motion for $f(r)$ can be integrated into an algebraic equation, so \eqref{eq:Zmyers} provides an example of a QTG. It is not trivial for $D \geq 5$ and was the first non-trivial and non-Lovelock instance of a GQTG ever given in the literature.
\end{example}

\begin{example}
Let us now write the following density \cite{Hennigar:2017ego}:
\begin{align}
\nonumber
\mathcal{S}_{(3)}^{\mathrm{HKM}}&=+14 R_{a\ b}^{\ c \ d}R_{c\ d}^{\ e \ f}R_{e\ f}^{\ a \ b}+2 R_{abcd} R^{abc}_{\ \ \ e} R^{de}-\frac{(38-29D+4 D^2)}{4(D-2)(2D-1)}R_{abcd} R^{abcd} R\\ \label{eq:sdgqtg} &- \frac{2(4D^2+9D-30)}{(D-2)(2D-1)}R_{abcd} R^{ac} R^{bd}- \frac{4(2D^2-35D+66)}{3(D-2)(2D-1)} R_{a}^{b}R_{b}^{c}R_{c}^{a}\\ \nonumber & +\frac{(4 D^2-21 D+34)}{(D-2)(2D-1)}R_{ab} R^{ab} R-\frac{(4D^2-13D+30)}{12(D-2)(2D-1)}R^3\,.
\end{align}
This theory is seen to fulfill the requirement \eqref{eq:gqtgdef} when evaluated on the SSS ansatz \eqref{eq:sssans} with $N(r)=1$, showing that it forms part of the GQTG family. The equation of motion for $f(r)$ can be integrated into an equation of second order in derivatives, so \eqref{eq:sdgqtg} defines a proper GQTG. For $D=4$, $\mathcal{S}^{\mathrm{HKM}}$ is equivalent (as a GQTG) to the following theory:
\begin{equation}
\mathcal{P}=12 R_{a\ b}^{\ c \ d}R_{c\ d}^{\ e \ f}R_{e\ f}^{\ a \ b}+R_{ab}{}^{cd}R_{cd}{}^{ef}R_{ef}{}^{ab}-12R_{abcd}R^{ac}R^{bd}+8R_{a}^{b}R_{b}^{c}R_{c}^{a}\, ,
\label{eq:ECG}
\end{equation}
The theory \eqref{eq:ECG} was first identified in \cite{Bueno:2016lrh} and called Einstenian Cubic Gravity. This was the first example of proper GQTG ever provided in the literature.
\end{example}

The previous examples are of cubic order in the curvature, but GQTGs of higher-order are known. On the one hand, recurrent relations to construct QTGs and one single example of a proper GQTG at every order have been found \cite{Bueno:2019ycr}. For the benefit of the reader, let us present them here:
\begin{align}
\label{eq:zrecu}
&\hspace{-0.2cm}\Z_{(n+5)}^{\mathrm{BCH}} =\frac{3(n+3) \Z_{(1)}^{\mathrm{BCH}} \Z_{(n+4)}^{\mathrm{BCH}}}{D(1-D)(n+1)}+\frac{3(n+4) \Z_{(2)}^{\mathrm{BCH}} \Z_{(n+3)}^{\mathrm{BCH}}}{D(D-1)n}-\frac{(n+3)(n+4)\Z_{(3)}^{\mathrm{BCH}} \Z_{(n+2)}^{\mathrm{BCH}}}{n D(D-1)(n+1)}\,, \\ & \cS_{(n+5)}^{\mathrm{BCH}}=\frac{3(n+3) \cS_{(1)}^{\mathrm{BCH}} \cS_{(n+4)}^{\mathrm{BCH}}}{4(1-D)(n+1)}+\frac{3(n+4) \cS_{(2)}^{\mathrm{BCH}} \cS_{(n+3)}^{\mathrm{BCH}}}{4(D-1)n}-\frac{(n+3)(n+4)\cS_{(3)}^{\mathrm{BCH}} \cS_{(n+2)}^{\mathrm{BCH}}}{4n(D-1)(n+1)}\,,
\label{eq:srecu}
\end{align}
where the recurrent relations begin with the densities $\mathcal{Z}_{(m)}^{\mathrm{BCH}}$ and $\mathcal{S}_{(m)}^{\mathrm{BCH}}$ with $m=1,2,3,4,5$ as given in \cite{Bueno:2019ycr}. Amusingly, equations \eqref{eq:zrecu} and \eqref{eq:srecu} are formally equivalent, just differing by an innocent prefactor. On the other hand, in this Ref. an explicit formula for a QTG at each order $n$ in $D \geq 5$ and the expression for a proper GQTG at every curvature order in $D \geq 4$ were derived. However, these formulae are somewhat involved and it would be desirable to have at disposal more manageable and more direct expressions for them.  

Furthermore, it turns out that the covariant form (i.e. not evaluated in any particular ansatz) for a representative of each equivalence class of GQTGs at every order $n$ and for generic $D \geq 5$ has not been obtained yet in the literature. Although their expression when evaluated on the reduced SSS ansatz \eqref{eq:sssans} with $N(r)=1$ is known \cite{Bueno:2022res}, an explicit covariant characterization of all inequivalent GQTGs is clearly of relevance when it comes to the study of gravitational configurations that do not possess such static and spherical symmetry. Consequently, we devote ourselves to the finding  of covariant expressions for representatives of each equivalence class of GQTGs at every order $n$ in the curvature.

\section{Reconstructing off-shell densities from on-shell densities in $D \geq 5$}\label{sec:offon}

Let us consider the reduced SSS ansatz \eqref{eq:sssans} with $N(r)=1$:
\begin{equation}
\diff s_{f}^2=- f(r) \diff t^2+ \frac{\diff r^2}{f(r)}+r^2 \diff \Omega^2_{D-2}\,.
\label{eq:redsss}
\end{equation}
Define  
\begin{equation}
\label{eq:defabc}
A=\frac{f''(r)}{2}\, , \quad B = -\frac{f'(r)}{2r}\, , \quad \psi=\frac{1-f(r)}{r^2}\,.
\end{equation}
For every $D \geq 5$, we may always pick \cite{Bueno:2022res} GQTG representatives $\Z_{(n)}$ and $S_{(n,j)}$, with $j=2,...,n-1$, whose evaluation on the reduced SSS ansatz \eqref{eq:redsss} reads:\footnote{Notice that we employed a basis for $\mathcal{S}_{(n,j)}$ different from the one presented in Equation (37) in \cite{Bueno:2022res}  (see published version, in v1 of arXiv there seems to be a small typo in the coefficient of their $\tau_{(n,j+1)}$, which should read $(j-1)(Dj-4n)$), since their $\tau_{(n,j)}$ turn out not to provide a set of $n-2$ inequivalent GQTG for certain values of $n$ and $D$. An explicit example of this issue can be seen for $n=5$ and $D=5$, after noticing that the subsequent proper GQTGs labeled by $j=3$ ans $j=4$ coincide identically.}
\begin{align}
\label{eq:zos}
\Z_{(n)} \vert_f&=\frac{1}{r^{D-2}} \frac{\diff}{\diff r}\left [ r^{D-1}\left( (2n-D)\tau_{(n,0)}-2n\tau_{(n,1)}\right) \right]\,,\\ \label{eq:sos}
\cS_{(n,j)} \vert_f&=\frac{1}{r^{D-2}} \frac{\diff}{\diff r}\left [ r^{D-1}\left (\left(2-\frac{D}{2n}(j+1)\right)\tau_{(n,0)}-(j+1)\tau_{(n,j)}+(j-1)\tau_{(n,j+1)}\right) \right],
\end{align}
where $n\geq 1$ for $\Z_{(n)}\vert_f$ and $j=2,\dots,n-1$ with $n \geq 3$ for $\cS_{(n,j)}\vert_f$. We have also defined
\begin{equation}
\tau_{(n,k)}\equiv\psi^{n-k} B^{k}\,,\quad k=0, \dots,n\,.
\end{equation}
Now, it is natural to wonder if it is possible to find an algorithm with which to \emph{uplift} the \emph{on-shell}\footnote{Throughout this document, \emph{on-shell} refers to ``evaluated on the particular SSS ansatz \eqref{eq:redsss}''.} expressions \eqref{eq:zos} and \eqref{eq:sos} into fully covariant or \emph{off-shell} expressions. This way, we would achieve to find a representative for every equivalence class of GQTGs at every order in the curvature.

To tackle the problem, we could try to find a \emph{dictionary} which allows to relate any on-shell expression $\Sigma(r)$ into a covariant combination of curvature invariants $\mathcal{R}$ whose evaluation on the background \eqref{eq:redsss} precisely yields\footnote{Note that such dictionary can never be one-to-one --- indeed, there are many different combinations of curvature invariants that produce the same expression when evaluated on \eqref{eq:redsss}.} $\R\vert_f=\Sigma(r)$. To this aim, we first notice that for $D \geq 4$ the Riemann tensor, when evaluated on \eqref{eq:redsss}, can be written as \cite{Deser:2005pc}
\begin{equation}
\left. R^{a b}{}_{c d} \right \vert_{f}= 2 \left[-A T_{[c}^{[a} T_{d]}^{b]}+2 B T_{[c}^{[a} \sigma_{d]}^{b]}+\psi \sigma_{[c}^{[a} \sigma_{d]}^{b]} \right]\,, 
\end{equation}
where $T^a_b=\delta^a_t \delta_b^t+\delta^a_r \delta_b^r $ and $\sigma^a_b=\delta^a_b-T^a_b$ are the projectors into the $(t,r)$ and angular directions, respectively.\footnote{Observe that $T_a^b T_b^c=T_a^c$, $\sigma_a^b \sigma_b^c=\sigma^a_c$, $T_a^b \sigma_b^c=0$, $T_a^a=2$ and $\sigma^a_a=D-2$.} In turn, these implies that:
\begin{align}
\left. R^a_b \right\vert_f&=((D-2) B-A) T^a_b+(2B+(D-3) \psi) \sigma^a_b,\\ 
 \left.R\right\vert_f&=4(D-2) B-2A+(D-3) (D-2) \psi.
\end{align}
Similarly, the Weyl tensor $W_{abcd}$ and the traceless Ricci tensor $Z_{ab}
$, which are the proper objects --- together with the Ricci scalar and the metric --- appearing in the Ricci decomposition of the Riemann tensor, read:
\begin{align}
\label{eq:weylsim}
\left.W^{ab}{}_{cd}\right\vert_f&=\Omega(r) \left[ \frac{(D-2)(D-3)}{2}T_{[c}^{[a} T_{d]}^{b]}-(D-3) T_{[c}^{[a} \sigma_{d]}^{b]}+ \sigma_{[c}^{[a} \sigma_{d]}^{b]} \right]\, , \\ 
\left.Z^a_b\right\vert_f&=\Theta(r)\left[ -\frac{D-2}{2} T^a_b+ \sigma^a_b \right ]\, ,
\label{eq:ricim}
\end{align}
where
\begin{align}
\label{eq:omsim}
\Omega(r)&=\frac{4-4f(r)+4 r f'(r)-2r^2 f''(r)}{(D-1)(D-2) r^2}\, , \\
\label{eq:thetasim} \Theta(r)&=\frac{2(D-3)(1-f(r))+(D-4) r f'(r)
+r^2 f''(r)}{D r^2}\,.
\end{align}
For completeness, we also present $\Rho(r)=R\vert_f$ in terms of $f(r)$ and its derivatives:
\begin{equation}
\Rho(r)=\frac{(D-2)((D-3)(1-f(r))-2r f'(r))-r^2 f''(r)}{r^2}\,.
\label{eq:scalsim}
\end{equation}
Now we continue the following proposition, part of which was already proven in \cite{Deser:2005pc,Bueno:2019ltp}.
\begin{proposition}
\label{prop:wzr}
Let $m$, $p$ and $q$ be positive integers such that $m^2+p^2>2$ and assume $D \geq 4$. Then, if $(W^m Z^p R^q)_i$ stands for a generic curvature invariant constructed from $m$ Weyl tensors, $p$ traceless Ricci tensors and $q$ Ricci scalars,
\begin{equation}
\left.(W^m Z^p R^q)_i\right\vert_f=c_i \Omega^m \Theta^p \Rho^q\,, 
\end{equation}
where $c_i$ is a certain numerical coefficient that depends on the specific invariant considered. 
\end{proposition}
\begin{proof}
Follows directly by observation of Equations \eqref{eq:weylsim}, \eqref{eq:ricim}, \eqref{eq:omsim} and \eqref{eq:thetasim}. Indeed, the contraction of the different projectors appearing in these equations will just give numerical contributions, encoded in $c_i$.
\end{proof}
\noindent
Note that if $m=p=0$ in the previous proposition, the statement is trivial, while the case $0<m^2+p^2 \leq 2$ is not considered since there are no curvature invariants built out with just one Weyl tensor, just one traceless Ricci tensor or with one Weyl tensor and one traceless Ricci tensor.

Let us now define the following densities for\footnote{The case $D=4$ is treated separately in Section \ref{sec:gqtgs4}.} $D\geq 5$:
\begin{align}
\label{eq:d1}
\sW_2 &\equiv\frac{4}{(D-2)^2(D-1)(D-3)} W_{a b c d} W^{a b c d}\,, \\
\label{eq:d2}
\sZ_2 & \equiv\frac{2}{D(D-2)} Z^a_b Z_a^b\, ,\\\
\label{eq:d3}
\sW_3 &\equiv \frac{8}{(D-3)(D-2)(2(2-(D-3)^2)+(D-2)^2(D-3)^2)} W\indices{_a_b^c^d}W\indices{_c_d^e^f}W\indices{_e_f^a^b}\,,\\
\label{eq:d4}
\sY_3 & \equiv \frac{8}{D^2(D-2)(D-3)} Z^a_b Z^c_d W\indices{_a_c^b^d}\, , \\
\label{eq:d5}
\sX_3 & \equiv -\frac{8}{(D-1)^2(D-2)(D-3)(D-4)} Z^a_b W_{a c d e}W^{bcde}\,,\\
\label{eq:d6}
\sZ_3 & \equiv -\frac{4}{D(D-2)(D-4)} Z^a_b Z^b_cZ_a^c\, , \\
\label{eq:d7}
\sY_4 & \equiv - \frac{16}{D^2(D-2)(D-3)(D-4)}  Z^{a}_b Z_{a c} Z_{d e} W^{bdce} \, ,\\
\label{eq:d8}
\sX_4 & \equiv -\frac{32}{D(D-1)^2(D-2)(D-3)^2(D-4)} Z^{ab} W_{acbd} W^{c efg} W^d{}_{efg}\, .
\end{align}
It turns out that:
\begin{equation}
\begin{split}
\left. \sW_2 \right \vert_f &= \Omega^2\, , \quad \left. \sZ_2 \right \vert_f= \Theta^2\, , \quad \left. \sW_3 \right \vert_f= \Omega^3\, , \quad \left. \sY_3 \right \vert_f= \Theta^2 \Omega\, ,  \\    \sX_3  \vert_f &=  \Omega^2\Theta\, , \quad \left. \sZ_3 \right \vert_f = \Theta^3\, , \quad  \left. \sY_4 \right \vert_f= \Theta^3 \Omega\,, \quad \left. \sX_4 \right \vert_f= \Omega^3 \Theta\,.
\end{split}
\end{equation}
\begin{proposition}
\label{teo:dsss}
Let $\R_{(n)}$ be any curvature density of arbitrary order $n$. Then there exist certain coefficients $a_i$ and non-negative integers $b_{m}^{(i)}$, with $m \in \{1,...,9\}$ and $i \in \{1,2,...,k_n\}$ for certain positive integer $k_n$, such that
\begin{equation}\label{eq:dsss}
\left. \R_{(n)} \right \vert_f= \sum_{i=1}^{k_n} a_i \left. \left ( R^{b_{1}^{(i)}} \sW_2^{b_{2}^{(i)}} \sZ_2^{b_{3}^{(i)}} \sW_3^{b_{4}^{(i)}}\sY_3^{b_{5}^{(i)}} \sX_3^{b_{6}^{(i)}} \sZ_3^{b_{7}^{(i)}} \sY_4^{b_{8}^{(i)}} \sX_4^{b_{9}^{(i)}} \right)\right \vert_f\,.
\end{equation}
\end{proposition}
\begin{proof}
Given any $\R_{(n)}$, let us replace all its Riemann and Ricci tensors by Weyl tensors, traceless Ricci tensors and Ricci scalars. This way, $\R_{(n)}$ is written as a linear combination of densities with the structure $(W^m Z^p R^q)$. By Proposition \ref{prop:wzr}, it turns out that all densities of the form $(W^m Z^p R^q)$, when evaluated on the reduced SSS ansatz \eqref{eq:redsss}, are proportional to $\Omega^m \Theta^p \Rho^q$. Since it always holds that either $m=p=0$ or $m^2+p^2>2$ (see explanation below Proposition \ref{prop:wzr}), it is possible to express $\R_{(n)} \vert_f$  in terms of the densities in \Cref{eq:d1,eq:d2,eq:d3,eq:d4,eq:d5,eq:d6,eq:d7,eq:d8} as in \eqref{eq:dsss} (for appropriate choices of $a_i$ and $b_{m}^{(i)}$) and we conclude. 
\end{proof}

Proposition \ref{teo:dsss} is the key result that will allow us to translate on-shell expressions into proper covariant or off-shell densities. Indeed, if we define:
\begin{align}
\I^{(1)}_l&=\sW_2^{\frac{l-\pi_l}{2}} \left ( (1-\pi_l) \sW_2 +\pi_l \sW_3\right)\, , \\
\I^{(2)}_l&=\sZ_2^{\frac{l-\pi_l}{2}} \left (  (1-\pi_l) \sZ_2 +  \pi_l \sZ_3 \right)\,, \\
\I^{(3)}_l &=\sW_2^{\frac{l-\pi_l}{2}} \left ( (1-\pi_l) \sX_3 +  \pi_l \sX_4\right)\, , \\
\I^{(4)}_l&=\sZ_2^{\frac{l-\pi_l}{2}} \left (  (1-\pi_l) \sY_3 +  \pi_l \sY_4 \right)\,,
\end{align}
where
\begin{equation}
\pi_l= \left \lbrace \begin{matrix}
0\,  \quad l \,\, \mathrm{even}\\ 
1\,  \quad l \,\, \mathrm{odd}
\end{matrix} \right.\,,
\end{equation}
we can construct the following \emph{dictionary} relating\footnote{It is not a bijective relation, of course.} on-shell quantities to off-shell covariant densities,\footnote{Note that there are different choices of $a_i$ and $b_{m}^{(i)}$ in Proposition \ref{teo:dsss} that produce the same expression $\left. \R_{(n)} \right \vert_f$. Therefore, the dictionary we present here is just a convenient instance of the many different dictionaries that one may construct.} with $l \geq 0$: 
\begin{align}
\nonumber
\textbf{On-shell} & \hspace{1.3cm} \textbf{Off-shell} \\ \label{eq:dic1}
\Omega(r)^{l+2} & \xrightarrow{\hspace*{1.3cm}} \I_l^{(1)}\\  \label{eq:dic2}
\Theta(r)^{l+2} &\xrightarrow{\hspace*{1.3cm}}   \I_l^{(2)}\\ \label{eq:dic3}
\Theta(r) \Omega(r)^{l+2} &\xrightarrow{\hspace*{1.3cm}}  \I_l^{(3)}\\ \label{eq:dic4}
\Omega(r) \Theta(r)^{l+2} &\xrightarrow{\hspace*{1.3cm}}  \I_l^{(4)} \\
\label{eq:dic5}
\Rho(r) &\xrightarrow{\hspace*{1.3cm}}  R\,.
\end{align}

\noindent
Equipped with this dictionary, the procedure to follow to obtain a covariant fixed expression for a representative of each equivalence class of GQTGs at all curvature orders is:
\begin{enumerate}
\item Start with the on-shell expressions \eqref{eq:zos} and \eqref{eq:sos}. 
\item Massage the corresponding on-shell forms to write everything in terms of $\Rho(r)$, $\Omega(r)$ and $\Theta(r)$.
\item Apply the dictionary (\Cref{eq:dic1,eq:dic2,eq:dic3,eq:dic4,eq:dic5}) to obtain off-shell expressions whose evaluation on \eqref{eq:redsss} precisely yields the initial on-shell quantity.
\end{enumerate}

\noindent The algorithm is clear, but there is a subtle issue to address. Indeed, observe that there are \emph{a priori} three types of terms that are not covered by the previous dictionary: $\Rho^q \Theta$, $\Rho^q \Omega$ and $\Rho^q \Theta \Omega$. Nevertheless, those terms can never come from a covariant off-shell expression, since no curvature invariants may be formed with just one Weyl tensor, just one traceless Ricci tensor or one Weyl tensor and one traceless Ricci tensor. Therefore, for a potential on-shell GQTG to arise from a true off-shell combination of densities, such terms cannot appear.


\section{All Generalized Quasitopological Gravities}\label{sec:all}


Once presented the method to convert on-shell expressions into off-shell densities, we are in position of finding a representative of each equivalence class of GQTGs. We will start finding a representative of the unique QTG that exists at each order in curvature (for $D \geq 5$), and then we will continue with the characterization of inequivalent proper GQTGs in $D \geq 5$, commenting also on the number of curvature invariants that are needed to write them.



\subsection{Quasitopological Gravities}

Let us follow the procedure described in Section \ref{sec:offon}. We start with Eq. \eqref{eq:zos}, which we write in terms of $A$, $B$ and $\psi$ as defined in \eqref{eq:defabc} \cite{Bueno:2019ycr}:
\begin{equation}
\label{eq:qtgos}
\Z_{(n)} \vert_f=-4n(n-1) B^2 \psi^{n-2}+n(2A-4(D-2n)B) \psi^{n-1}-(D-2n)(D-2n-1) \psi^n\,,
\end{equation}
where $n \geq 1$. The second step of the procedure requires to express $\Z_{(n)} \vert_f$ in terms of $\Rho$, $\Omega$ and $\Theta$. For that, we note that:
\begin{align}
\label{eq:abca1}
A &=-\frac{1}{D(D-1)} \Rho+ \Theta-\frac{(D-2)(D-3)}{4}\Omega\, , \\
\label{eq:abca2}
B &=\frac{1}{D(D-1)} \Rho- \frac{D-4}{2(D-2)} \Theta-\frac{D-3}{4}\Omega\, , \\\label{eq:abca3}
\psi &=\frac{1}{D(D-1)} \Rho+ \frac{2}{D-2}\Theta+\frac{1}{2}\Omega\, .
\end{align}
Substituting in \eqref{eq:qtgos} and massaging carefully the subsequent expression, a somewhat involved calculation and direct use of the dictionary (\Cref{eq:dic1,eq:dic2,eq:dic3,eq:dic4,eq:dic5}) produces the unique inequivalent covariant QTG\footnote{We normalize the coefficient of $R^n$ to one.} existing at each curvature order for $D\geq 5$.
\begin{theorem}
\label{theo:qtgs}
A representative of the unique equivalence class of QTGs existing at each curvature order $n \geq 3$ for $D \geq 5$ can be chosen to be\footnote{We remind the usual convention that whenever the upper limit of a summation is bigger than the lower limit, such summation is identically zero.}
\begin{align}
\Z_{(n)}&=R^n+ \sum_{l=0}^{n-2}  R^{n-l-2}\left (\gamma_{n,-2,l}\I^{(1)}_l+\gamma_{n,l,-2}\I^{(2)}_l\right)+\sum_{l=0}^{n-3}  R^{n-l-3}\left (\gamma_{n,-1,l} \I^{(3)}_l+\gamma_{n,l,-1}\I^{(4)}_l\right)\notag \\ &+ \sum_{l=0}^{n-4}\sum_{p=0}^{n-l-4} \gamma_{n,l,p} R^{n-l-p-4}\I^{(1)}_p \I^{(2)}_l\,, \quad n \geq 3\, ,\label{eq:QGnWZR}
\end{align}
where the constants $\gamma_{n,l,p}$ are only non-zero for $l,p \geq -2$ and $l+p +4 \leq n$, in which case
\begin{align}
\label{eq:gamma}
\gamma_{n,l,p}&=\frac{n!(D(D(l-2)+4)(l+1)+4(D-1)(D l+1)(p+2)+4(D-1)^2(p+2)^2)}{2^{2-l+p} (D^2-D)^{-p-l-3} (D-2)^{l+2} (l+2)!(p+2)!(n-l-p-4)! }\,,
\end{align}
\end{theorem}
\noindent
Let us present here the explicit form of $\mathcal{Z}_{(n)}$ from $n=1$ to $n=4$:
\begin{align}
\Z_{(1)}&=R\, , \\
\Z_{(2)}&=R^2+\frac{(D-1) D W_{a b c d} W^{a b c d}}{(D-3) (D-2)}-\frac{4 (D-1) D Z^a_b Z_a^b}{(D-2)^2}= \frac{D(D-1)}{(D-2)(D-3)}\mathcal{X}_4\,, \\
\mathcal{Z}_{(3)}&=R^3+\frac{2 (D-1)^2 D^2 (2 D-3) W\indices{_a_b^c^d}W\indices{_c_d^e^f}W\indices{_e_f^a^b}}{(D-3) (D-2) (D ((D-9)
   D+26)-22)}+\frac{24 (D-1)^2 D^2 Z^a_b Z^c_d W\indices{_a_c^b^d}}{(D-3) (D-2)^3}\notag\\
   &+\frac{16 (D-1)^2
   D^2 Z^a_b Z^b_cZ_a^c}{(D-2)^4}-\frac{24 (D-1)^2 D^2 Z^a_b W_{a c d e}W^{bcde}}{(D-4) (D-3)
   (D-2)^2}+\frac{3 (D-1) D R W_{a b c d} W^{a b c d}}{(D-3) (D-2)}\notag\\
   &-\frac{12 (D-1) D R
   Z^a_b Z_a^b}{(D-2)^2}\, ,\\
\mathcal{Z}_{(4)}&=R^4+\frac{3 (D-1)^2 D^3 (3 D-4) \left(W_{a b c d} W^{a b c d}\right)^2}{(D-3)^2 (D-2)^4}-\frac{384 (D-1)^3 D^3
   Z^{a}_b Z_{a c} Z_{d e} W^{bdce}}{(D-4) (D-3) (D-2)^4}\notag\\
   &+\frac{8 (D-1)^2 D^2 (2 D-3) R
   W\indices{_a_b^c^d}W\indices{_c_d^e^f}W\indices{_e_f^a^b}}{(D-3) (D-2) (D ((D-9) D+26)-22)}+\frac{96 (D-1)^2 D^2 R
   Z^a_b Z^c_d W\indices{_a_c^b^d}}{(D-3) (D-2)^3}\notag\\
   &+\frac{64 (D-1)^2 D^2 R Z^a_b Z^b_cZ_a^c}{(D-2)^4}-\frac{96 (D-1)^2
   D^2 R Z^a_b W_{a c d e}W^{bcde}}{(D-4) (D-3) (D-2)^2}\notag\\
   &+\frac{24 (D-1)^2 D^2 (D (7
   D-10)+4) W_{a b c d} W^{a b c d}Z^e_f Z_e^f}{(D-3) (D-2)^5}+\frac{192 (D-1)^3 D^2
   \left(Z^a_b Z_a^b\right)^2}{(D-2)^6}\notag\\
   &-\frac{192 (D-1)^3 D^2 Z^{ab} W_{acbd} W^{c efg} W^d{}_{efg}}{(D-4) \left(D^2-5
   D+6\right)^2}+\frac{6 (D-1) D R^2 W_{a b c d} W^{a b c d}}{(D-3) (D-2)}\notag\\
   &-\frac{24 (D-1) D R^2
   Z^a_b Z_a^b}{(D-2)^2}\,,
\end{align}
where $\mathcal{X}_4$ is the Gauss-Bonnet density. Observe that $\mathcal{Z}_{(3)}$ is equivalent (as a QTG) to the cubic theory presented in Eq. \eqref{eq:Zmyers}, since
\begin{equation}
\Z_{(3)}=\frac{4 (D-1)^2 D^2 (2 D-3)}{(D-3) (D-2) (D ((D-9) D+26)-22)}\left ( \mathcal{Z}^{\text{MR}}_{(3)}+\frac{1}{8}\X_6 \right) \, ,
\end{equation}
where $\X_6$ is the cubic Lovelock density, given by
\begin{align}
\mathcal{X}_6&=-8\tensor{R}{_{a}^{c}_{b}^{d}}\tensor{R}{_{c}^{e}_{d}^{f}}\tensor{R}{_{e}^{a}_{f}^{b}}+4\tensor{R}{_{ab}^{cd}}\tensor{R}{_{cd}^{ef}}\tensor{R}{_{ef}^{ab}} -24\tensor{R}{_{abcd}}\tensor{R}{^{a bc}_{e}}R^{d e} \notag \\
&+3\tensor{R}{_{abcd}}\tensor{R}{^{abcd}}R+24\tensor{R}{_{abcd}}\tensor{R}{^{ac}}\tensor{R}{^{bd}}
+16 R^{b}_{a} R_{b}^{c} R_{c}^{a}
-12R_{ab}R^{ab} R+R^3\, .
\end{align}
For the benefit of the reader, we also show in the appendix the explicit form of the QTG densities \eqref{eq:QGnWZR} for $n=5$ and $n=6$.

Theorem \ref{theo:qtgs} provides a representative of the unique equivalence class of QTGs that exists at each curvature order in $D \geq 5$. Remarkably, Eq. \eqref{eq:QGnWZR} turns out to be dramatically simpler than the explicit QTGs at all orders and dimensions given in \cite{Bueno:2019ycr}, which we have verified to differ from our QTGs by trivial GQTGs. Additionally, we have explicitly checked that the Lagrangian \eqref{eq:QGnWZR} satisfies the recurrence relation \eqref{eq:zrecu} up to trivial densities --- that is, the on-shell evaluation of \eqref{eq:QGnWZR} satisfies \eqref{eq:zrecu} exactly.





The QTGs \eqref{eq:QGnWZR} are defined for $D \geq 5$. In the case $D=4$, we explained before that there seems to be no QTGs (apart from Einstein gravity), which will be rigorously proven\footnote{We remind that, for $D=3$,  the only densities satisfying the GQTG condition \eqref{eq:gqtgdef} are trivial ones (except for the Einstein-Hilbert term, of course) \cite{Bueno:2022lhf}.} in Section \ref{sec:gqtgs4}. 




\subsection{Generalized Quasitopological Gravities in $D \geq 5$}

Again, we have to apply the procedure explained in Section \ref{sec:offon}. We start by writing the on-shell expressions for the $n-2$ inequivalent GQTG representatives in \eqref{eq:sos} in terms of the variables $A,B$ and $\psi$ defined in \eqref{eq:defabc}:
\begin{align}\notag
\mathcal{S}_{(n,j)}\vert_f&=\frac{(D+Dj-4n)(2n-D+1)}{2n}\psi^n-(D+Dj-4n)\psi^{n-1}B+j(j+1)A\psi^{n-j}B^{j-1}\label{eq:gqgos}\\
&+(j+1)(2n-D-j+1)\psi^{n-j}B^j-(j^2-1)A\psi^{n-j-1}B^j\\\notag
&+(D(j-1)+j(1+3j-4n))\psi^{n-j-1}B^{j+1}-2(j-1)(1+j-n)\psi^{n-j-2}B^{j+2}\,,
\end{align}
with $n \geq 3$ and $j=2,\dots,n-1$. The following step requires to express $A, B$ and $\psi$ in terms of $\Rho$, $\Omega$ and $\Theta$, which may be done by use of Eqs. \eqref{eq:abca1}, \eqref{eq:abca2} and \eqref{eq:abca3}. Then, after some involved computations, one can write the subsequent expression for \eqref{eq:gqgos} in a form which is readily adapted for application of our on-shell to off-shell dictionary given by \Cref{eq:dic1,eq:dic2,eq:dic3,eq:dic4,eq:dic5}. This way, one gets the covariant form for the $n-2$ GQTGs existing\footnote{We normalize the coefficient of the term $R^n$ to one.} at each curvature order $n$ in $D \geq 5$.
\begin{theorem}
\label{theo:gqtgs}
Representatives of each of the $n-2$ equivalence classes of proper GQTGs existing at each curvature order $n\geq 3$ with $D \geq 5$ can be taken to be
\begin{align}
\label{eq:gGQTGoff}
&\hspace{-0.3cm }\mathcal{S}_{(n,j)} =  R^n+ \sum_{l=0}^{n-2}  R^{n-l-2}\left (\sigma_{n,j,-2,l} \I^{(1)}_l+\sigma_{n,j,l,-2}\I^{(2)}_l\right)\\&\hspace{-0.2cm }+\sum_{l=0}^{n-3}  R^{n-l-3}\left (\sigma_{n,j,-1,l} \I^{(3)}_l+\sigma_{n,j,l,-1}\I^{(4)}_l\right)+ \sum_{l=0}^{n-4}\sum_{p=0}^{n-l-4} \sigma_{n,j,l,p} R^{n-l-p-4}\I^{(1)}_p \I^{(2)}_l\,, \quad n \geq 3\,, \nonumber
\end{align}
where $j=2,\dots,n-1$ and where the constants $\sigma_{n,j,l,p}$ with $l,p \geq -2$ are given by
\begin{align}
\sigma_{n,j,l,p}&=\frac{\rho_{n,j,l+2,p+2}}{D^2-D}+\frac{(D-3)(D-2)}{4} \rho_{n,j,l+2,p+1}-\rho_{n,j,l+1,p+2}-\frac{\nu_{n,j,l+2,p+2}}{j+1}\,, \\ 
\rho_{n,j,l,p}&=2n(D^2-D)^{n-1} (j\mu_{n-j,j-1,l,p}-(j-1)\mu_{n-j-1,j,l,p}) \\ 
\nu_{n,j,l,p}&=(D^2-D)^{n-1}\Big [ (D+D j-4n)(2n+1-D)\mu_{n,0,l,p}-2n(D+D j-4n)\mu_{n-1,1,l,p} \nonumber\\
&-2n(j+1)(D-1+j-2n)\mu_{n-j,j,l,p}+2n(D(j-1)+j(1+3j-4n))\mu_{n-j-1,j+1,l,p}\nonumber\\
& -4n(j-1)(j+1-n)\mu_{n-j-2,j+2,l,p} \Big]\,,\\
\mu_{n,j,l,p}&=\sum_{k=0}^l \eta_{n,j,l,p,k} P_p^{(j+k-l-p,l-n-j-1)}\left ( \frac{D+1}{D-3} \right) \,,
\end{align}
where $P_n^{(\alpha,\beta)}(x)$ stands for the Jacobi polynomial and where the coefficients $\eta_{n,j,l,p,k}$ for $j\geq l-k$ are given by
\begin{equation}
\eta_{n,j,l,p,k}=\left (\frac{3-D}{4} \right)^{p} \left (\frac{2}{D-2} \right)^l \left (\frac{4-D}{4} \right)^{l-k} (D^2-D)^{l+p-n-j} \begin{pmatrix}
n \\ k
\end{pmatrix}\begin{pmatrix}
j \\ l-k
\end{pmatrix}\,,
\end{equation}
and are zero otherwise.
\end{theorem} 
Despite the quite intricate expression for the coefficients $\sigma_{n,j,l,p}$, we note they are just numbers for each dimension $D \geq 5$. Therefore, we have explicitly found, for the first time, the covariant form for a representative of all equivalence classes of GQTGs in $D \geq 5$. We have checked that a linear combination of such inequivalent GQTGs happens to be identical to the explicit GQTG found at each curvature order in \cite{Bueno:2019ycr}, up to the addition of trivial GQTGs. 

For $n=3,4$, the formula \eqref{eq:gGQTGoff} produces the following GQTGs:
\begin{align}
\mathcal{S}_{(3,2)}&=R^3-\frac{(D-1)^2 D^2 (D ((D-3) D-5)+11) W\indices{_a_b^c^d}W\indices{_c_d^e^f}W\indices{_e_f^a^b}}{2 (D-3) (D-2) (D ((D-9)
   D+26)-22)}-\frac{12 (D-1) D R Z^a_b Z_a^b}{(D-2)^2}\notag\\
   &-\frac{6 (D-5) (D-1)^2 D^2 Z^a_b Z^c_d W\indices{_a_c^b^d}}{(D-3) (D-2)^3}-\frac{(D-4) (D-1)^2 D^2
   (D+4) Z^a_b Z^b_cZ_a^c}{(D-2)^4}\notag\\
   &+\frac{3 (D-1)^2 D^2 (D+4) Z^a_b W_{a c d e}W^{bcde}}{2 (D-3) (D-2)^2}+\frac{3 (D-1)
   D R W_{a b c d} W^{a b c d}}{(D-3) (D-2)}\, .
\end{align}
\begin{align}
\mathcal{S}_{(4,2)}&=R^4+\frac{8(D-1)^2 D^3 (D-3) (D-1) (D (D+4)-52) Z^{a}_b Z_{a c} Z_{d e} W^{bdce}}{(D-4) (D-3)^2 (D-2)^4}\notag\\
   &-\frac{(D-1)^2 D^3(D-4) (D ((D-3)
   D-6)+11) \left(W_{a b c d} W^{a b c d}\right)^2}{(D-4) (D-3)^2 (D-2)^4}\notag\\
   &-\frac{2 (D-1)^2 D^2 (D ((D-3)
   D-21)+35) R W\indices{_a_b^c^d}W\indices{_c_d^e^f}W\indices{_e_f^a^b}}{3 (D-3) (D-2) (D ((D-9) D+26)-22)}\notag\\
   &-\frac{8 (D-13) (D-1)^2 D^2 R
   Z^a_b Z^c_d W\indices{_a_c^b^d}}{(D-3) (D-2)^3}-\frac{4 (D-1)^2 D^2 \left(D^2-48\right) R Z^a_b Z^b_cZ_a^c}{3 (D-2)^4}\notag\\
   &+\frac{2 (D-1)^2
   D^2 \left(D^2-48\right) R Z^a_b W_{a c d e}W^{bcde}}{(D-4) (D-3) (D-2)^2}+\frac{4 (D-1)^3 D^2 \left((D-8)
   D^2+144\right) \left(Z^a_b Z_a^b\right)^2}{3 (D-2)^6}\notag\\
   &-\frac{4 (D-1)^2 D^2 (D (D
   (D (2 D-3)-41)+60)-24) W_{a b c d} W^{a b c d}Z^e_f Z_e^f}{(D-3) (D-2)^5}\notag\\
   &+\frac{4 (D-1)^3 D^2 (D (11 D-12)-140) Z^{ab} W_{acbd} W^{c efg} W^d{}_{efg}}{3 (D-4)
   \left(D-2\right)^2 \left(D-3\right)^2}+\frac{6 (D-1) D R^2 W_{a b c d} W^{a b c d}}{(D-3) (D-2)}\notag\\
   &-\frac{24 (D-1) D R^2
   Z^a_b Z_a^b}{(D-2)^2}\, ,
\end{align}
\begin{align}
\mathcal{S}_{(4,3)}&=R^4+\frac{3 (D-1)^2 D^3 (D (D ((D-8) D+18)-4)-11) \left(W_{a b c d} W^{a b c d}\right)^2}{4 (D-3)^2 (D-2)^4}\notag\\
   &+\frac{120
   (D-1)^3 D^3 Z^{a}_b Z_{a c} Z_{d e} W^{bdce}}{(D-3) (D-2)^4}-\frac{24 (D-5) (D-1)^2 D^2 R Z^a_b Z^c_d W\indices{_a_c^b^d}}{(D-3)
   (D-2)^3}\notag\\
   &-\frac{2 (D-1)^2 D^2 (D ((D-3) D-5)+11) R
   W\indices{_a_b^c^d}W\indices{_c_d^e^f}W\indices{_e_f^a^b}}{(D-3) (D-2) (D ((D-9) D+26)-22)}\notag\\
   &-\frac{4 (D-4) (D-1)^2 D^2 (D+4) R Z^a_b Z^b_cZ_a^c}{(D-2)^4}+\frac{6 (D-1)^2 D^2 (D+4) R
   Z^a_b W_{a c d e}W^{bcde}}{(D-3) (D-2)^2}\notag\\
   &+\frac{3 (D-1)^2 D^2 (D (D (D ((D-9) D+11)+53)-80)+32) W_{a b c d} W^{a b c d}
   Z^a_b Z_a^b}{(D-3) (D-2)^5}\notag\\
   &-\frac{4 (D-1)^3 D^2 (2 (D-3) D-11) Z^{ab} W_{acbd} W^{c efg} W^d{}_{efg}}{\left(D-2\right)^2 \left(D-3\right)^2}+\frac{6
   (D-1) D R^2 W_{a b c d} W^{a b c d}}{(D-3) (D-2)}\notag\\
   &-\frac{(D-1)^3 D^2 \left(D^2 ((D-9) D+32)-192\right)
   \left(Z^a_b Z_a^b\right)^2}{(D-2)^6}-\frac{24 (D-1) D R^2 Z^a_b Z_a^b}{(D-2)^2}\, .
\end{align}
The density $\mathcal{S}_{(3,2)}$ can be related to the expression \eqref{eq:sdgqtg} as follows:
\begin{align}
\mathcal{S}_{(3,2)}&=\frac{3 (D-1)^2 D^2 (2 D-1)}{8 (D-3) (15-D (D+5))}\mathcal{S}_{(3)}^{\text{HKM}}\notag\\
&-\frac{2 (D-1)^2 D^2 (D-4) (2 D-3)
   (D (D (19 D-141)+220)-22) }{8 (D-3) (15-D (D+5)) (D-2) (D ((D-9) D+26)-22)}\Z_{(3)}^{\text{MR}}\notag\\
   &-\frac{(D-1)^2 D^2  (D ((D-3) D-5)+11)}{8 (D-3) (D-2)  (D ((D-9) D+26)-22)}\X_6\, .
\end{align}

In Appendix \ref{sec:App2} we provide the explicit covariant expressions for representatives of all equivalence classes of GQTGs existing at $n=5$ and $n=6$ for $D \geq 5$. In another vein, we have checked that a specific combination of the $n-1$ inequivalent GQTGs --- this is, proper GQTG \eqref{eq:gGQTGoff} and QTG \eqref{eq:QGnWZR} densities --- existing at each order $n$ satisfies the recurrence relation \eqref{eq:srecu} up to trivial GQTG densities.

\subsubsection{Number of curvature invariants in GQTGs for $D \geq 5$}\label{sec:numberD}

Let us consider a certain scalar off-shell density $\mathcal{R}_{n}$ of the form:
\begin{equation}
\R_{(n)}= \sum_{i=1}^{k_n} a_i  \left ( R^{b_{1}} \sW_2^{b_{2}} \sZ_2^{b_{3}} \sW_3^{b_{4}}\sY_3^{b_{5}} \sX_3^{b_{6}} \sZ_3^{b_{7}} \sY_4^{b_{8}} \sX_4^{b_{9}} \right)\,.
\label{eq:basiscurvature}
\end{equation}
Its curvature order is given by $n=b_1+2(b_2+b_3)+3(b_4+b_5+b_6+b_7)+4(b_8+b_9)$. Following this observation, we are now interested in answering the following question: What is the number of off-shell densities $\#_n^{\text{(G)QTG}}$ of the form \eqref{eq:basiscurvature} that appear after direct application of the dictionary (Eqs. \eqref{eq:dic1} to \eqref{eq:dic5}) to the on-shell expressions \eqref{eq:qtgos} and \eqref{eq:gqgos}? Naively, one would expect $\#_n^{\text{(G)QTG}}$ to coincide with the number $\#_n$ of off-shell terms that may be constructed from products of the densities in Eqs. \eqref{eq:d1} to \eqref{eq:d8} (and Ricci scalars). We can easily see that this is not the case by inspecting the situation up to fifth order in curvature. The lowest five-order densities are given by:
\begin{align}
&d_1=\{R\}\,,\\
&d_2=R\cdot d_1\cup\left\{\sW_2,\sZ_2\right\}\, ,\\
&d_3=R\cdot d_2\cup\left\{\sW_3,\sX_3,\sY_3,\sZ_3\right\}\, ,\\
&d_4=R\cdot d_3\cup\left\{\right.\sW_2^2,\sZ_2^2,\sX_4,\sY_4\left.\right\}\, ,\\
&d_5=R\cdot d_4\cup\left\{\sW_2\sW_3,\sW_2\sX_3, \sW_2\sY_3,\sW_2\sZ_3,\sZ_2\sW_3,\sZ_2\sX_3,\sZ_2\sY_3,\sZ_2\sZ_3\right\}\,,
\end{align}
where the notation $R \cdot d_i$ means that all elements of $d_i$ are multiplied by $R$. We observe that $\#_1=1$, $\#_2=3$, $\#_3=7$, $\#_4=12$ and $\#_5=20$ respectively. However, when studying their on-shell evaluation, at order 5 there exist two equivalent ways of obtaining the on-shell quantities $\Omega^3\Theta^2$ and $\Omega^2\Theta^3$:
\begin{align}
\Omega^3\Theta^2&=\sW_3\sZ_2\big|_f=\sW_2\sY_3\big|_f,\\
\Omega^2\Theta^3&=\sW_2\sZ_3\big|_f=\sX_3\sZ_2\big|_f.
\end{align} 
This implies that among the $\#_5=20$ possible densities existing at $n=5$, two of them are not independent when they are evaluated on the SSS ansatz \eqref{eq:redsss}. As a consequence, since we are using the specific dictionary provided by Eqs. \eqref{eq:dic1} to \eqref{eq:dic5}, the off-shell expression of the QTG and GQTG following \eqref{eq:QGnWZR} and \eqref{eq:gGQTGoff} will include only $\#_5^{\text{(G)QTG}}=18$ terms. Another example of this degeneracy is given at $n=6$, as the combinations $\sW_3^2\big|_f=\sW_2^3\big|_f$ and $\sZ_3^2\big|_f=\sZ_2^3\big|_f$ yield the same on-shell result.

It is possible to find explicitly both $\#_n$ and $\#^{\text{(G)QTG}}_n$ by following the same procedures of \cite{Bueno:2022lhf,Bueno:2022ewf}. There, they used the technique of generating functions, which is well known in number theory within the problem of integer partitions. Regarding $\#_n$, we have been able to obtain it using the generating function
\begin{equation}
\mathcal{G}(x)=\frac{1}{1-x}\frac{1}{(1-x^2)^2}\frac{1}{(1-x^3)^4}\frac{1}{(1-x^4)^2}\, ,
\end{equation}
where $x$ indicates the curvature order $n$. Its Maclaurin expansion reads
\begin{equation}
\mathcal{G}(x)=\sum_{n}\#_nx^n=1+x+3x^2+7x^3+12x^4+20x^5+38x^6+58x^7+\mathcal{O}(x^8)\, .
\end{equation}
The coefficient of the monomial of degree $n$ gives precisely $\#_n$. These coefficients are observed to satisfy the following recurrence relation for any $n$:
\begin{equation}
(n+25)\#_n+(n+27)\#_{n+1}+(n+33)\#_{n+2}+19\#_{n+3}-(n-2)\#_{n+4}-(n+4)\#_{n+5}-(n+6)\#_{n+6}=0\, ,
\end{equation}
Now, we look for a generating function $\mathcal{G}^{\text{(G)QTG}}(x)$ for $\#^{\text{(G)QTG}}_n$. As such, $\mathcal{G}^{\text{(G)QTG}}(x)$ quantifies the number of off-shell densities that the (G)QTG involves at a certain order in curvature $n$. By inspection, we find that
\begin{equation}
\mathcal{G}^{\text{(G)QTG}}(x)=1+\frac{x^4-x^3-x}{(x-1)^3}\, ,
\end{equation}
renders the desired result. In this case the expansion reads
\begin{equation}
\mathcal{G}^{\text{(G)QTG}}(x)=\sum_{n}\#_n^{\text{(G)QTG}}x^n=1+x+3x^2+7x^3+12x^4+18x^5+25x^6+33x^7+\mathcal{O}(x^8)\, .
\end{equation}
We have been able to find a closed expression for $\#^{\text{(G)QTG}}_n$. The result reads:
\begin{equation}
\#_{1}^{\text{(G)QTG}}=1\, , \quad \#_{n}^{\text{(G)QTG}}=\frac{(n-1)\left(n+4\right)}{2}\, , \quad n>1\, .
\label{eq:numbergqtgs}
\end{equation}
Interestingly enough, $\#_{n}^{\text{(G)QTG}}$ matches with the number of curvature invariants that are needed to produce the QTGs of \cite{Bueno:2019ycr}. In another vein, in this Ref. an instance of proper GQTG at each curvature order $n$ is presented as well, explaining that they contain $2n-1$ curvature invariants. We argue that the fact that such number of curvature invariants grows linearly with $n$ instead of quadratically as in \eqref{eq:numbergqtgs} is due to the particular GQTG considered\footnote{Of course, for $D \geq 5$, it will correspond to a particular linear combination of the GQTGs $\mathcal{S}_{(n,j)}$ of Eq. \eqref{eq:gGQTGoff}, up to trivial GQTGs.} at each curvature order in \cite{Bueno:2019ycr}. This was chosen to connect smoothly with the $D=4$ case, for which a unique inequivalent GQTG exists (as we will prove in next section). In fact, remarkably, $2n-1$ is precisely the number $\#_{n}^{\text{(G)QTG,(4)}}$  of terms that are necessary to write the four-dimensional GQTGs presented in Eq. \eqref{eq:gqtg4}, as we will comment afterwards.

\section{Generalized Quasitopological Gravities in $D=4$}
\label{sec:gqtgs4}

The algorithm developed in previous section works for $D \geq 5$. Nevertheless, it fails for $D=4$, since the
generic on-shell analysis carried out in \cite{Bueno:2022res} cannot be applied. Despite this, they argued that there only exists a unique inequivalent GQTG at each derivative level. It is the purpose of this section to prove rigorously this statement and construct explicitly a representative of such unique equivalence class of GQTGs at every curvature order. 


For that, we need to adapt the generic algorithm of Section \ref{sec:offon} to $D=4$. First, we observe that Proposition \ref{prop:wzr} is valid as well for $D=4$. Secondly, we note the following very special property of curvature invariants in $D=4$.
\begin{proposition}
\label{prop:noricodd4}
Let us consider a curvature invariant purely built out from Weyl tensors, Ricci scalars and an odd number of traceless Ricci tensors in $D=4$ (schematically, it would of the form $W^m Z^p R^q$, with $m,p$ and $q$ non-negative integers such that $p$ is odd and $m \neq 1$). Then it is identically zero when evaluated on the ansatz \eqref{eq:redsss}.
\end{proposition}
\begin{proof}
The expressions of $W_{ab}{}^{cd}$ and $Z_{a}^b$ in terms of the projectors $T^a_b$ and $\sigma^a_b$ into the $(t,r)$ and angular directions are obtained by setting $D=4$ in Eqs. \eqref{eq:weylsim} and \eqref{eq:ricim}:
\begin{align}
\label{eq:weylsimd4}
W^{ab}{}_{cd}&=\Omega(r) \left[ T_{[c}^{[a} T_{d]}^{b]}- T_{[c}^{[a} \sigma_{d]}^{b]}+ \sigma_{[c}^{[a} \sigma_{d]}^{b]} \right]\,, \\ 
Z^a_b&=\Theta(r)\left[ -T^a_b+ \sigma^a_b \right ]\,,
\label{eq:ricimd4}
\end{align}
where $\Omega(r)$ and $\Theta(r)$ are given by Eqs. \eqref{eq:omsim} and \eqref{eq:thetasim} after imposing $D=4$. Define the linear map $\mathfrak{f}$ on the space of $(1,1)$ tensors whose application on the projectors  $T^a_b$ and $\sigma^a_b$ reads
\begin{equation}
 \mathfrak{f}(T^a_b)=\sigma^a_b \, , \quad  \mathfrak{f}(\sigma^a_b)= T^a_b\,.
 \label{eq:substproof}
\end{equation}
In particular, $\mathfrak{f}(\delta_a^b)=\delta_a^b$. Given a tensor product of $n$ projectors $\mathcal{A}^{(1)}_{\ a_1 b_1}\mathcal{A}^{(2)}_{\ a_2 b_2} \cdots \mathcal{A}^{(n)}_{\ a_n b_n}$, with each $\mathcal{A}^{(k)}$ being either $T$ or $\sigma$, we can extend naturally the definition  of $\mathfrak{f}$ as follows:
\begin{equation}
\mathfrak{f}\left ( \mathcal{A}^{(1)}_{\ a_1 b_1} \dots \mathcal{A}^{(n)}_{\ a_n b_n}\right )=\mathfrak{f}\left ( \mathcal{A}^{(1)}_{\ a_1 b_1} \right) \dots \mathfrak{f}\left ( \mathcal{A}^{(n)}_{\ a_n b_n} \right)\,.
\end{equation}
This map commutes with the contraction of projectors, since $\mathfrak{f}(\mathcal{A}^a_c \mathcal{B}^c_b )=\mathfrak{f}(\mathcal{A}^a_c)\mathfrak{f}(\mathcal{B}^c_b)$, with $\mathcal{A}$, $\mathcal{B}$ being either $T$ or $\sigma$. We have that
\begin{equation}
\mathfrak{f}\left (\left. W^{ab}{}_{cd} \right \vert_f \right)=\left. W^{ab}{}_{cd} \right \vert_f\, , \quad\mathfrak{f}(\left. Z^a_b \right \vert_f )=-\left. Z^a_b\right \vert_f\,.
\label{eq:iwz}
\end{equation}
Now, let us consider a rank-two tensor $K^a_b$ obtained through an arbitrary contraction of $m$ Weyl tensors and $p$ traceless Ricci tensors. Schematically:
\begin{equation}
K^a_b= W_{c_1 d_1}{}^{e_1 f_1} \dots W_{c_m d_m}{}^{e_m f_m} Z^{g_1}_{h_1} \dots Z^{g_p}_{h_p} \, \Xi^{c_1 d_1 \dots c_m d_m h_1 \dots h_p a}_{e_1 f_1 \dots e_m f_m g_1 \dots g_p b}\,,
\label{eq:ktermswz}
\end{equation}
where the tensor $\Xi$ is made up of Kronecker deltas that ensure the contraction of the different Weyls and traceless Riccis. Observe that any curvature invariant constructed purely from $m$ Weyl tensors, $p$ traceless Ricci tensors and $q$ Ricci scalars may be expressed as $R^q K_a^a$ for certain integer $q \geq 0$. Since the Weyl and traceless Ricci tensors are built through the projectors $T^a_b$ and $\sigma^a_b$ when evaluated on the reduced SSS ansatz \eqref{eq:redsss}, then $\left. K^a_b \right \vert_f$ is also purely expressed in terms of these projectors:
\begin{equation}
\left. K^a_b \right \vert_f= s_1 T^a_b+ s_2 \sigma^a_b\,,
\label{eq:ktermsproy}
\end{equation}
for certain functions $s_1$ and $s_2$. However, taking into account the definition of $K^a_b$ in Eq. \eqref{eq:ktermswz} and using Eq. \eqref{eq:iwz}:
\begin{equation}
\mathfrak{f}(K^a_b)=(-1)^p K^a_b\,.
\end{equation}
Therefore, we learn that $s_1=(-1)^p s_2$. Specifically, if $p$ is odd (that is, there is an odd number of traceless Ricci tensors), we have that $K^a_b$ is actually proportional to $Z^a_b$ and, in particular, traceless. Taking into account that any curvature invariant constructed purely from Weyl tensors, traceless Riccis and Ricci scalars may be expressed as $R^q K_a^a$ for appropriate integer $q>0$ and tensor $K_b^a$, we conclude.
\end{proof}
The previous proposition modifies drastically the situation.\footnote{Proposition \ref{prop:noricodd4} only holds for $D=4$ since it is the unique dimension for which $\left. Z_{ab} \right \vert_f$ is an eigenvector of the linear map $\mathfrak{f}$ defined in Eq. \eqref{eq:substproof}. For $D>4$, observe that $\mathfrak{f}\left (\left. Z_{ab} \right \vert_f \right)$ (after extending canonically the definition of $\mathfrak{f}$ for any dimension $D>4$) is no longer proportional to $\left. Z_{ab} \right \vert_f$.} Indeed, it implies that no terms of the form $\Rho^m \Omega^n \Theta^{2k+1}$ may come from an actual covariant density. 

Define the following densities in $D=4$:
\begin{align}
\label{eq:d1n}
\sW_2 &=\frac{1}{3} W_{a b c d} W^{a b c d}\,, \\
\label{eq:d2n}
\sZ_2 & = \frac{1}{4} Z^{ab} Z_{ab}\, ,\\\
\label{eq:d3n}
\sW_3 &= \frac{2}{3} W_{a b c d} W^{c d e f} W_{ef}{}^{ab}\,,\\
\label{eq:d4n}
\sY_3 & = \frac{1}{4} Z^{ab} Z^{cd} W_{acbd}\, , 
\end{align}
\begin{proposition}
Let $\mathcal{R}_{(n)}$ be any curvature density of arbitrary order $n$ in $D=4$. Then there exists certain coefficients $a_i$ and non-negative integers $b_m^{(i)}$, with $m \in {1,\dots 5}$ and $i \in {1,2,\dots, k_n}$ for certain positive integer $k_n$ such that
\begin{equation}\label{eq:prop4}
\left. \mathcal{R}_{(n)} \right \vert_f= \sum_{i=1}^{k_n} a_i \left.  \left ( R^{b_1^{(i)}} \sW_2^{b_2^{(i)}} \sZ_2^{b_3^{(i)}} \sW_3^{b_4^{(i)}} \sY_3^{b_5^{(i)}} \right) \right \vert_f \,.
\end{equation}
\end{proposition}
\begin{proof}
Goes along similar lines to that of Proposition \ref{teo:dsss}, taking into account Proposition \ref{prop:noricodd4}.
\end{proof}
The corresponding on-shell to off-shell dictionary takes a simpler form in $D=4$. Indeed, now it just reads:
\begin{align}
\nonumber
\textbf{On-shell} & \hspace{1.3cm} \textbf{Off-shell} \\ \label{eq:dic1d4}
\Omega(r)^{l+2} & \xrightarrow{\hspace*{1.3cm}} \I_l^{(1)}\\  \label{eq:dic2d4}
\Theta(r)^{l+2} &\xrightarrow{\hspace*{1.3cm}} (1-\pi_l) \sZ_2^{1+\frac{l-\pi_l}{2}}  \\ \label{eq:dic3d4}
\Omega(r)\Theta(r)^{l+2} &\xrightarrow{\hspace*{1.3cm}} (1-\pi_l)\sZ_2^{\frac{l-\pi_l}{2}} \sY_3 \,,\\ \label{eq:dic4d4}
\Rho(r) &\xrightarrow{\hspace*{1.3cm}}  R\,,
\end{align}
with $l \geq 0$ and where $\I_l^{(1)}$ is as given in \eqref{eq:dic1}. Also, note that $\Omega(r)$, $\Theta(r)$ and $\Rho(r)$ are obtained by taking $D=4$ in \eqref{eq:omsim}, \eqref{eq:thetasim} and \eqref{eq:scalsim}. 

At this point, we have to note another difference with respect to the case $D \geq 5$. In that situation, we could rely on the results of \cite{Bueno:2022res}, in which it was explicitly proven that there only exist $n-1$ inequivalent GQTGs at each curvature order, providing as well their on-shell form. For $D=4$, their argument does not work, and it is not true that there are exactly $n-1$ equivalence classes in four dimensions. We may understand this from Proposition \ref{prop:noricodd4} and from the four-dimensional dictionary (\Cref{eq:dic1d4,eq:dic2d4,eq:dic3d4,eq:dic4d4}), since there is no way by which terms $\Rho^m \Omega^n \Theta^{2k+1}$ come from an actual covariant density. This reduces the number of inequivalent GQTGs at each curvature order to be one, as we proceed to show now.



Let us consider the on-shell expressions $\left. \mathcal{Z}_{(n)} \right \vert_{f}$ and $\left. \mathcal{S}_{(n,j)} \right \vert_{f}$ presented in Eqs. \eqref{eq:zos} and \eqref{eq:sos}. Now, they will not generically correspond to actual GQTG densities, but let us maintain the notation for the sake of simplicity. Consider
\begin{equation}
\mathcal{F}_n=\alpha_n \mathcal{Z}_{(n)}\vert_f+ \sum_{j=2}^{n-1}\beta_{n,j} \mathcal{S}_{(n,j)}\vert_f\,, \quad n \geq 3\,,
\label{eq:eqFons}
\end{equation}
for arbitrary coefficients $\alpha_n, \beta_{n,j}$. The question is: do there exist choices of $\alpha_n$ and $\beta_{n,j}$ which guarantee that $\mathcal{F}_n$ contains no terms of the form $\Rho^m \Omega^n \Theta^{2k+1}$? Otherwise, $\mathcal{F}_n$ would not arise from an actual covariant density. Note that we impose $n \geq 3$, since for $n=1$ and $n=2$ the corresponding $\mathcal{F}_n$ can be trivially seen to circumvent these issues.

It turns out that there exists a unique choice for $\alpha_n, \beta_{n,j}$ that prevents the presence of terms $\Rho^m \Omega^n \Theta^{2k+1}$. 
\begin{proposition}
\label{prop:onshellfos}
Let $\mathcal{F}_n$ to be as in \eqref{eq:eqFons}. It corresponds to the on-shell evaluation of an actual covariant curvature invariant if and only if
\begin{equation}
\alpha_n=\beta_{n,j}=0\,, \quad j=2,\dots,n-2\,.
\label{eq:condsufnec}
\end{equation} 
\end{proposition}
\begin{proof}

Observe that
\begin{align}
\left. \frac{\diff \mathcal{F}_3}{\diff \Theta} \right \vert_{\Theta=0}&=-\frac{27}{2}  \alpha_3\Omega^2 \,, \\
\left. \frac{\diff \mathcal{F}_4}{\diff \Theta} \right \vert_{\Theta=0}&=-\frac{9}{8}  (4\alpha_4+\beta_{4,2})\Rho \Omega^2-27 \left (\alpha_4+\frac{1}{32}\beta_{4,2}\right )\Omega^3 \,, \\ \nonumber
\left. \frac{\diff \mathcal{F}_5}{\diff \Theta} \right \vert_{\Theta=0}&=-\left (\frac{45}{4} \alpha_5+\frac{117}{64}\beta_{5,2}+\frac{9}{32} \beta_{5,3}\right )\Rho \Omega^3+\frac{15}{64}\left (4\alpha_5-\beta_{5,2}-\frac{4}{5} \beta_{5,3}\right )\Rho^2 \Omega^2\,, \\ &-\frac{27}{128}(160 \alpha_5+6 \beta_{5,2}+9 \beta_{5,3})\,, \\
\label{eq:fonshell4}
\left. \frac{\diff \mathcal{F}_n}{\diff \Theta} \right \vert_{\Theta=0}&= \sum_{l=0}^{n-1} h_{n,l} Q^{n-l-1} B^l\,,  \qquad n \geq 6\,,
\end{align} 
where $B=\frac{1}{12}\Rho-\frac{1}{4}\Omega$ (it corresponds to taking $D=4$ in \eqref{eq:abca2}), $Q=\frac{1}{12}\Rho+\frac{1}{2}\Omega$ and:
\begin{align}
h_{n,0}&=-4n(n-1)(n-2) \alpha_n +2(2n-3)\sum_{j=2}^{n-2} (1+j-n) \beta_{n,j} \, , \\
h_{n,1}&=8n(n-1)(n-2) \alpha_n+6(3-n) \beta_{n,2} -4(n-1)\sum_{j=2}^{n-2} (1+j-n) \beta_{n,j} \, , \\
h_{n,2}&=-4n(n-1)(n-2) \alpha_n +6(n-1)(n-3) \beta_{n,2}+12(4-n) \beta_{n,3} \, , \\
h_{n,3}&=(18-8n)(n-3) \beta_{n,2}+8(n-1)(n-4) \beta_{n,3}+20(5-n) \beta_{n,4} \, , \\
h_{n,l}&=2(n-l)(l-3)(n-l+1) \beta_{n,l-2}+(n-l)(4n+(3l-4n-1)l-6)\beta_{n,l-1} \\ \nonumber &+ (l+1)(l+2)(l+2-n) \beta_{n,l+1}+2(l+1)(n-1)(n-l-1) \beta_{n,l} \, , \quad l=4,\dots,n-2\,,  \\
h_{n,n-1}&=4(n-4) \beta_{n,n-3}-(n^2-n+2) \beta_{n,n-2}\,.
\end{align}
Assume that $\mathcal{F}_n$ comes from an actual covariant density. Then, it is necessary that $\left. \frac{\diff \mathcal{F}_n}{\diff \Theta} \right \vert_{\Theta=0}=0$. For the cases $3 \leq n \leq 5$, we directly derive Eq. \eqref{eq:condsufnec} by inspection of the previous equations. For $n \geq 6$, the vanishing of \eqref{eq:fonshell4} implies that $h_{n,k}=0$, with $k=0,\dots,n-1$. This produces a linear system of $n$ equations for the $n-2$ variables $(\alpha,\beta_{n,2},\beta_{n,3},\dots, \beta_{n,n-2})$. Among them, one may study in further detail the subsystem of $n-3$ equations posed by $h_{n,l}=0$, with $l=3,\dots,n-1$, whose variables are $(\beta_{n,2},\beta_{n,3},\dots, \beta_{n,n-2})$. Let $\mathbf{C}$ be the matrix of coefficients associated to such subsystem of equations. The determinant of $\mathbf{C}$ can be computed to be
\begin{equation}
\det \mathbf{C}=-(-1)^n 3^{n-3} \frac{\Gamma(n-1)\Gamma(n-2) \Gamma(n-2/3)}{\Gamma(7/3)}\,.
\end{equation}
The previous determinant is nonzero for every $n \geq 3$. Consequently, the subsystem of equations $h_{n,l}=0$ with $l=3,\dots,n-1$ has a unique solution that corresponds to $\beta_{n,2}=\beta_{n,3}=\dots=\beta_{n,n-2}=0$. Then, the remaining equations are satisfied if and only if  $\alpha_n=0$, so that Eq. \eqref{eq:condsufnec} is obtained. 

Conversely, assume that $\alpha_n=\beta_{n,j}=0$, for $j=2,\dots,n-2$ and $n \geq 3$. Then, application of the dictionary Eqs. \eqref{eq:dic1d4} to \eqref{eq:dic4d4} provides the off-shell density presented below in Eq. \eqref{eq:gqtg4} and we conclude. 
\end{proof}

The previous proposition states that the only choice for $\mathcal{F}_n$ that comes from an actual covariant density is
\begin{equation}
\mathcal{F}_n=\mathcal{S}_{(n,n-1)}\vert_f=\frac{3}{12^n}(\Rho-3 \Omega)^{n-2}\left (-2\Rho^2-6(n-2) \Rho \Omega+3(n-1) (16 n \Theta^2+3(2-3n) \Omega^2) \right )\,,
\label{eq:gqtgd4os}
\end{equation}
where we imposed $j=n-1$ and $D=4$ in \eqref{eq:gqgos} and simplified the resulting expression. For $n=1$ and $n=2$, we have that 
\begin{equation}
\mathcal{F}_1=-\frac{\Rho}{2}\, , \quad \mathcal{F}_2=-\frac{\Rho^2}{24}+2\Theta^2-\frac{3 \Omega^2}{4}\,.
\end{equation}
$\mathcal{F}_1$ trivially corresponds to the Einstein-Hilbert term (up to a constant), while $\mathcal{F}_2$ is the on-shell evaluation of a term proportional to the Gauss-Bonnet density in four dimensions, which is known to be topological. Therefore, to quadratic order we verify that the only four-dimensional GQTG is Einstein gravity. For $n \geq 3$, direct application of the dictionary (Eqs. \eqref{eq:dic1d4} to \eqref{eq:dic4d4}) in \eqref{eq:gqtgd4os} proves the following result, after normalizing the coefficient of the term $R^n$ to one.
\begin{theorem}
\label{theo:d4gqtgs}
There exists a unique inequivalent GQTG at each curvature order $n \geq 3$ in $D=4$. A representative of such unique equivalence class can be taken to be
\begin{equation}
\mathcal{S}^{(4)}_{(n)}= R^n+\gamma_1 R^{n-2} \sZ_2+ \gamma_2 R^{n-3} \sY_3+ \sum_{l=0}^{n-2} \lambda^{(1)}_{l} R^{n-l-4} \I_l^{(1)} \left ( R^2+\lambda_l^{(2)} \sZ_2 \right)\,,
\label{eq:gqtg4}
\end{equation}
where
\begin{align}
\gamma_1=-24n (n-1)\,, \quad \gamma_2&=-3 (n-2) \gamma_1 \,, \quad \lambda^{(1)}_l=\frac{(-3)^{l+2} (l+1)(3l+4) n!}{2 (l+2)!(n-l-2)!}\,, \\
\lambda_l^{(2)}&=-\frac{48 (n-l-2)(n-l-3)}{(l+1)(3l+4)}\,.
\label{eq:coefsgqtg4}
\end{align}
\end{theorem}
Theorem \ref{theo:d4gqtgs} and the argumentation above provide, to the best of our knowledge, the first rigorous proof of the fact that there exists one and only one inequivalent four-dimensional GQTG at each curvature order $n$. Theorem \ref{theo:d4gqtgs} also presents the explicit off-shell expression of a representative of the unique equivalence class existing at every $n$, which we have proven to coincide with the four-dimensional GQTGs found in \cite{Bueno:2019ycr} up to trivial GQTGs. Also, we have checked that \eqref{eq:gqtg4} satisfies the recurrence relation \eqref{eq:srecu} up to trivial GQTG densities.
Eq. \eqref{eq:gqtg4} provides the following four-dimensional GQTGs at orders $n=3, 4, 5$ and $6$:
\begin{equation}
\mathcal{S}^{(4)}_{(3)}=R^3+18 R W_{a b c d} W^{a b c d}-36 R Z^a_b Z_a^b-126 W\indices{_a_b^c^d}W\indices{_c_d^e^f}W\indices{_e_f^a^b}+108 Z^a_b Z^c_d W\indices{_a_c^b^d}\, ,
\end{equation} 
\begin{align}
\mathcal{S}^{(4)}_{(4)}&= R^4+36 R^2 W_{a b c d} W^{a b c d}-72 R^2 Z^a_b Z_a^b-504 R W\indices{_a_b^c^d}W\indices{_c_d^e^f}W\indices{_e_f^a^b}+432 R Z^a_b Z^c_d W\indices{_a_c^b^d}\notag\\
&+135 \left(W_{a b c d} W^{a b c d}\right)^2-216 W_{a b c d} W^{a b c d}Z^e_f Z_e^f\, ,
\end{align} 
\begin{align}
\mathcal{S}^{(4)}_{(5)}&=R^5+60 R^3 W_{a b c d} W^{a b c d}-120 R^3 Z^a_b Z_a^b-1260 R^2 W\indices{_a_b^c^d}W\indices{_c_d^e^f}W\indices{_e_f^a^b}\notag\\
&+1080 R^2 Z^a_b Z^c_d W\indices{_a_c^b^d}+675 R \left(W_{a b c d} W^{a b c d}\right)^2-1080 R W_{a b c d} W^{a b c d}Z^e_f Z_e^f\notag\\
&-1404 W_{a b c d} W^{a b c d}W\indices{_e_f^g^h}W\indices{_g_h^i^j}W\indices{_i_j^e^f}+2160 Z^a_b Z_a^bW\indices{_c_d^e^f}W\indices{_e_f^g^h}W\indices{_g_h^c^d}\, ,
\end{align} 
\begin{align}
\mathcal{S}^{(4)}_{(6)}&=R^6+90 R^4 W_{a b c d} W^{a b c d}-180 R^4 Z^a_b Z_a^b-2520 R^3 W\indices{_a_b^c^d}W\indices{_c_d^e^f}W\indices{_e_f^a^b}\notag\\
&+2160 R^3 Z^a_b Z^c_d W\indices{_a_c^b^d}+2025 R^2 \left(W_{a b c d} W^{a b c d}\right)^2-3240 R^2 W_{a b c d} W^{a b c d}Z^e_f Z_e^f\notag\\
&-8424 R W_{a b c d} W^{a b c d}W\indices{_e_f^g^h}W\indices{_g_h^i^j}W\indices{_i_j^e^f}+12960 R Z^a_b Z_a^bW\indices{_c_d^e^f}W\indices{_e_f^g^h}W\indices{_g_h^c^d}\notag\\
&+1080 \left(W_{a b c d} W^{a b c d}\right)^3-1620 \left(W_{a b c d} W^{a b c d}\right)^2Z^e_f Z_e^f\, ,
\end{align}
The cubic density $\mathcal{S}_{(3)}^{(4)}$ is related to Einsteinian Cubic Gravity \eqref{eq:ECG} via
\begin{equation}
\mathcal{S}_{(3)}^{(4)}=18(\mathcal{T}-\mathcal{P})\, ,
\end{equation}
where $\mathcal{T}$ is a trivial GQTG in $D=4$ which has the following expression:
\begin{align}
\mathcal{T}&=12R_{a\ b}^{\ c \ d}R_{c\ d}^{\ e \ f}R_{e\ f}^{\ a \ b} - 6 R_{ab}{}^{cd}R_{cd}{}^{ef}R_{ef}{}^{ab} + 42 R_{abcd} R^{abc}_{\ \ \ e} R^{de} - 6 R R_{abcd}R^{abcd}\notag\\&
 - 48 R^{ab}R^{cd}R_{abcd} - 28 R^a_bR^b_cR^c_d + 24 R R^{ab}R_ {ab} - 2 R^3\, .
\end{align}
On the other hand, regarding the number of curvature invariants $\#_n^{\text{GQTG}, (4)}$ appearing in \eqref{eq:gqtg4}, if one takes into account that there are five curvature invariants at $n=3$ and that the sum generates two new terms at each order $n$, it is immediate to deduce that
\begin{equation}
\#_n^{\text{GQTG}, (4)}=2n-1\, .
\end{equation}
As mentioned before, this is precisely the number of terms that need to be included in the specific GQTGs found in \cite{Bueno:2019ycr}.


\section{Conclusions}\label{sec:concl}

In this paper we have carried out the classification of all inequivalent GQTGs in dimension $D \geq 4$. This characterization has been possible thanks to the dictionary derived in Section \ref{sec:offon} for $D \geq 5$ and in Section \ref{sec:gqtgs4} for $D=4$, which allowed the conversion of the on-shell classification of GQTGs developed in \cite{Bueno:2022res} into a fully covariant one. On the one hand, we were able to provide the covariant form of a simple representative of the unique equivalence class of QTGs --- see Theorem \ref{theo:qtgs} --- and of the $n-2$ equivalence classes of GQTGs --- see Theorem \ref{theo:gqtgs} --- existing in $D \geq 5$ at every curvature order $n$. Until now, only somewhat cumbersome expressions for QTGs and for a single example of GQTG at each curvature order were known \cite{Bueno:2019ycr}. On the other hand, we managed to rigorously prove that there exists one and only one equivalence class of (proper) GQTGs in four dimensions and to obtain the off-shell expression of a representative of such unique class in Theorem \ref{theo:d4gqtgs}. This is the first explicit proof regarding the existence of a unique class of inequivalent GQTGs in $D=4$. Also, the four-dimensional GQTGs to arbitrary order in the curvature presented in Eq. \eqref{eq:gqtg4} are much simpler in form than those provided before in the literature, cf. \cite{Bueno:2019ycr}.


Regarding future directions, there are several possibilities worth exploring. First, it would be interesting to have a better understanding on the role played by trivial GQTGs, such as parity-breaking densities. Indeed, note that we have carried out the classification of inequivalent GQTGs, which precisely mods out trivial GQTGs, so a full characterization of trivial GQTGs order by order would be necessary in order to obtain all GQTGs. This aspect is intimately connected with the study of non-SSS solutions in GQTGs, since trivial GQTGs need not be vanishing on different ans\"atze and, as such, may affect properties of solutions. In particular, it would be useful to investigate the properties of cosmological solutions and rotating black-hole solutions in GQTGs, examining the effects of adding or removing trivial GQTGs. This program has already been initiated for particular GQTGs \cite{Arciniega:2018fxj,Arciniega:2018tnn,Cisterna:2018tgx,Cano:2019ozf,Adair:2020vso,Cano:2020oaa}, so it would be of interest to extend it in more generality. Also, taking into account that GQTGs span the set of effective gravitational theories \cite{Bueno:2019ltp} after considering the possibility of metric redefinitions, the present classification, together with a potential complete characterization of trivial GQTGs, could help to unveil the minimum amount of GQTGs to obtain a basis for the set of effective theories of gravity. 

On the other hand, the methods developed in this work could be of utility in other setups. For example, one could consider the set of higher-curvature gravities which admit a holographic $c$-theorem \cite{Sinha:2010ai,Paulos:2010ke}. These theories are characterized by admitting domain-wall type solutions with second-order equations of motion which interpolate between two asymptotically AdS regions representing the ultraviolet and infrared fixed points of the renormalization group flow \cite{Girardello:1998pd,Freedman:1999gp}. Having at disposal the form of such theories when evaluated on the domain-wall like ansatz \cite{Bueno:2022log}, one could wonder whether it is possible to apply similar techniques to the ones presented in our work to obtain the covariant expression of theories satisfying holographic $c$-theorems at arbitrary orders. 

In more generality, it is natural to ask as well which new avenues our classification of inequivalent GQTGs may open up within the holographic context.







\section*{Acknowledgments}

We thank Pablo Bueno, Pablo A. Cano and Robie Hennigar for useful discussions. The research of J. M. is partially supported by Israel Science Foundation, grant no. 1487/21. The work of \'A. J. M. has been supported by a postdoctoral grant from the Istituto Nazione di Fisica Nucleare, Bando 23590. \'A. J. M. also wishes to thank E. Gil and A. Murcia for their permanent support.

\appendix 

\section{QTG densities at $n=5$ and $n=6$}\label{sec:App1}

We present here the QTGs $\mathcal{Z}_{(n)}$ of curvature order $n=5$ and $n=6$ that are obtained via Theorem \ref{theo:qtgs} for arbitrary dimension $D \geq 5$. We write them in terms of the Weyl tensor $W_{abcd}$, traceless Ricci tensor $Z_{ab}$ and the Ricci scalar $R$: 

\begin{align}
\Z_{(5)}&=R^5+\frac{4 (D-1)^3 D^4 (4 D-5) W_{a b c d} W^{a b c d}W\indices{_e_f^g^h}W\indices{_g_h^i^j}W\indices{_i_j^e^f}}{(D-3)^2 (D-2)^3 (D ((D-9)
   D+26)-22)}\notag\\
   &-\frac{240 (D-1)^3 D^4 W_{a b c d} W^{a b c d}Z^e_f W_{e g h i}W^{fghi}}{(D-4) (D-3)^2 (D-2)^4}-\frac{40 (D-1) D R^3 Z^a_b Z_a^b}{(D-2)^2}\notag\\
   &-\frac{1920 (D-1)^3 D^3 R
   Z^{a}_b Z_{a c} Z_{d e} W^{bdce}}{(D-4) (D-3) (D-2)^4}+\frac{15
   (D-1)^2 D^3 (3 D-4) R \left(W_{a b c d} W^{a b c d}\right)^2}{(D-3)^2 (D-2)^4}\notag\\
   &-\frac{160 (D-1)^3 D^3 (D (11 D-12)+4)
   W_{a b c d} W^{a b c d}Z^e_f Z^f_gZ_e^g}{(D-4) (D-3) (D-2)^6}\notag\\
   &+\frac{40 (D-1)^3 D^3 (D (17
   D-28)+12) Z^a_b Z_a^bW\indices{_c_d^e^f}W\indices{_e_f^g^h}W\indices{_g_h^c^d}}{(D-3) (D-2)^4 (D ((D-9) D+26)-22)}\notag\\
   &+\frac{1920
   (D-1)^4 D^3 Z^a_b Z_a^bZ^c_d Z^e_f W\indices{_c_e^d^f}}{(D-3) (D-2)^6}-\frac{256 (D-1)^4 D^3 (D+4)
   Z^a_b Z_a^bZ^c_d Z^d_eZ_c^e}{(D-4) (D-2)^7}\notag\\
   &+\frac{20 (D-1)^2 D^2 (2 D-3) R^2 W\indices{_a_b^c^d}W\indices{_c_d^e^f}W\indices{_e_f^a^b}}{(D-3)
   (D-2) (D ((D-9) D+26)-22)}+\frac{240 (D-1)^2 D^2 R^2 Z^a_b Z^c_d W\indices{_a_c^b^d}}{(D-3)
   (D-2)^3}\notag\\
   &+\frac{160 (D-1)^2 D^2 R^2 Z^a_b Z^b_cZ_a^c}{(D-2)^4}-\frac{240 (D-1)^2 D^2 R^2
   Z^a_b W_{a c d e}W^{bcde}}{(D-4) (D-3) (D-2)^2}\notag\\
   &+\frac{120 (D-1)^2 D^2 (D (7 D-10)+4) R
   W_{a b c d} W^{a b c d}Z^e_f Z_e^f}{(D-3) (D-2)^5}+\frac{960 (D-1)^3 D^2 R \left(Z^a_b Z_a^b\right)^2}{(D-2)^6}\notag\\
   &-\frac{960
   (D-1)^3 D^2 R Z^{ab} W_{acbd} W^{c efg} W^d{}_{efg}}{(D-4) (D-3)^2 (D-2)^2}+\frac{10 (D-1) D R^3
   W_{a b c d} W^{a b c d}}{(D-3) (D-2)}\, .
\end{align}

\begin{align}
\Z_{(6)}&=R^6 +\frac{5 (D-1)^3 D^5 (5 D-6) \left(W_{a b c d} W^{a b c d}\right)^3}{(D-3)^3 (D-2)^6}\notag\\
   &+\frac{24 (D-1)^3 D^4 (4
   D-5) R W_{a b c d} W^{a b c d}W\indices{_e_f^g^h}W\indices{_g_h^i^j}W\indices{_i_j^e^f}}{(D-3)^2 (D-2)^3 (D ((D-9) D+26)-22)}\notag\\
   &-\frac{1440
   (D-1)^3 D^4 R W_{a b c d} W^{a b c d}Z^e_f W_{e g h i}W^{fghi}}{(D-4) (D-3)^2 (D-2)^4}\notag\\
   &+\frac{60 (D-1)^3
   D^4 (D (31 D-54)+24) \left(W_{a b c d} W^{a b c d}\right)^2Z^e_f Z_e^f}{(D-3)^2 (D-2)^7}\notag\\
   &-\frac{960 (D-1)^4
   D^4 W_{a b c d} W^{a b c d}Z^{ef} W_{egfh} W^{g ijk} W^h{}_{ijk}}{(D-4) (D-3)^3 (D-2)^4}\notag\\
   &-\frac{320 (D-1)^4 D^4 (D (23
   D-32)+12) W\indices{_a_b^c^d}W\indices{_c_d^e^f}W\indices{_e_f^a^b}Z^g_h Z^h_iZ_g^i}{(D-4) (D-3) (D-2)^5 (D ((D-9)
   D+26)-22)}\notag\\
   &-\frac{15360 (D-1)^5 D^4 Z^a_b Z_a^bZ^{c}_d Z_{c e} Z_{f g} W^{dfeg}}{(D-4) (D-3)
   (D-2)^7}-\frac{5760 (D-1)^3 D^3 R^2 Z^{a}_b Z_{a c} Z_{d e} W^{bdce}}{(D-4) (D-3) (D-2)^4}\notag\\
   &-\frac{960
   (D-1)^3 D^3 (D (11 D-12)+4) R W_{a b c d} W^{a b c d}Z^e_f Z^f_gZ_e^g}{(D-4) (D-3)
   (D-2)^6}\notag\\
   &+\frac{240 (D-1)^3 D^3 (D (17 D-28)+12) R Z_{ab} Z^{ab} W_{cdef} W^{efgh} W_{gh}{}^{cd}}{(D-3)
   (D-2)^4 (D ((D-9) D+26)-22)}\notag\\
   &+\frac{11520 (D-1)^4 D^3 R Z^a_b Z_a^bZ^c_d Z^e_f W\indices{_c_e^d^f}}{(D-3) (D-2)^6}-\frac{1536 (D-1)^4 D^3 (D+4) R Z^a_b Z_a^bZ^c_d Z^d_eZ_c^e}{(D-4)
   (D-2)^7}\notag\\
   &+\frac{960 (D-1)^4 D^3 (D (8 D-7)+2) W_{a b c d} W^{a b c d}\left(Z^e_f Z_e^f\right)^2}{(D-3)
   (D-2)^8}\notag\\
   &+\frac{1280 (D-1)^5 D^3 (D+2) \left(Z^a_b Z_a^b\right)^3}{(D-2)^9}+\frac{40 (D-1)^2
   D^2 (2 D-3) R^3 W\indices{_a_b^c^d}W\indices{_c_d^e^f}W\indices{_e_f^a^b}}{(D-3) (D-2) (D ((D-9) D+26)-22)}\notag\\
   &+\frac{480
   (D-1)^2 D^2 R^3 Z^a_b Z^c_d W\indices{_a_c^b^d}}{(D-3) (D-2)^3}+\frac{320 (D-1)^2 D^2 R^3
   Z^a_b Z^b_cZ_a^c}{(D-2)^4}\notag\\
   &-\frac{480 (D-1)^2 D^2 R^3 Z^a_b W_{a c d e}W^{bcde}}{(D-4) (D-3)
   (D-2)^2}\notag\\
   &+\frac{360 (D-1)^2 D^2 (D (7 D-10)+4) R^2 W_{a b c d} W^{a b c d}Z^e_f Z_e^f}{(D-3)
   (D-2)^5}\notag\\
   &-\frac{2880 (D-1)^3 D^2 R^2
   Z^{ab} W_{acbd} W^{c efg} W^d{}_{efg}}{(D-4) (D-3)^2 (D-2)^2}+\frac{2880 (D-1)^3 D^2 R^2 \left(Z^a_b Z_a^b\right)^2}{(D-2)^6}\notag\\
   &+\frac{45 D^3 (3 D-4) \left(D-1\right)^2 R^2 \left(W_{a b c d} W^{a b c d}\right)^2}{(D-3)^2 (D-2)^4}+\frac{15 (D-1) D R^4 W_{a b c d} W^{a b c d}}{(D-3)
   (D-2)}\notag\\
   &-\frac{60 (D-1) D R^4 Z^a_b Z_a^b}{(D-2)^2}\, .
   \end{align}

\section{GQTG densities at $n=5$ and $n=6$}\label{sec:App2}

Here, we present the GQTG densities $\mathcal{S}_{(n,j)}$ of Theorem \ref{theo:gqtgs} at fifth and sixth order for $D\geq5$ in terms of the Weyl tensor $W_{abcd}$, traceless Ricci tensor $Z_{ab}$ and the Ricci scalar $R$. At each order, there are $j=2,\ldots, n-1$ inequivalent densities --- this is, three at $n=5$ and four at $n=6$:

\begin{align}  
   \mathcal{S}_{(5,2)}&=R^5-\frac{(D-1)^3 D^4 (D (5 (D-3) D-33)+55) W_{a b c d} W^{a b c d}W\indices{_e_f^g^h}W\indices{_g_h^i^j}W\indices{_i_j^e^f}}{3 (D-3)^2 (D-2)^3 (D ((D-9)
   D+26)-22)}\notag\\
   &+\frac{5 (D-1)^3 D^4 (D (23 D-32)-276) W_{a b c d} W^{a b c d}Z^e_f W_{e g h i}W^{fghi}}{6 (D-4) (D-3)^2
   (D-2)^4}\notag\\
   &-\frac{5 (D-1)^2 D^3 (D ((D-3) D-15)+23) R \left(W_{a b c d} W^{a b c d}\right)^2}{2 (D-3)^2 (D-2)^4}\notag\\
   &+\frac{20
   (D-1)^3 D^3 (D (D+4)-100) R Z^{a}_b Z_{a c} Z_{d e} W^{bdce}}{(D-4) (D-3) (D-2)^4}\notag\\
   &+\frac{5 (D-1)^3 D^3
   (D (D (D (43 D-52)-1044)+1152)-384) W_{a b c d} W^{a b c d}Z^e_f Z^f_gZ_e^g}{3 (D-4) (D-3) (D-2)^6}\notag\\
   &-\frac{5 (D-1)^3
   D^3 (D (D (D (24 D-61)-351)+648)-284) W\indices{_a_b^c^d}W\indices{_c_d^e^f}W\indices{_e_f^a^b} Z^a_b Z_a^b}{3 (D-3) (D-2)^4 (D ((D-9)
   D+26)-22)}\notag\\
   &-\frac{20 (D-1)^4 D^3 (D (7 D+12)-300) Z^a_b Z_a^bZ^c_d Z^e_f W\indices{_c_e^d^f}}{3 (D-3) (D-2)^6}\notag\\
   &+\frac{32
   (D-1)^4 D^3 (D (5 D-24)-96) Z^a_b Z_a^bZ^c_d Z^d_eZ_c^e}{3 (D-4) (D-2)^7}\notag\\
   &-\frac{5 (D-1)^2 D^2 (D
   ((D-3) D-45)+71) R^2 W\indices{_a_b^c^d}W\indices{_c_d^e^f}W\indices{_e_f^a^b}}{6 (D-3) (D-2) (D ((D-9) D+26)-22)}\notag\\
   &-\frac{10 (D-25) (D-1)^2
   D^2 R^2 Z^a_b Z^c_d W\indices{_a_c^b^d}}{(D-3) (D-2)^3}-\frac{5 (D-1)^2 D^2 \left(D^2-96\right) R^2 Z^a_b Z^b_cZ_a^c}{3
   (D-2)^4}\notag\\
   &-\frac{10
   (D-1)^2 D^2 (D (D (D (2 D-3)-83)+120)-48) R W_{a b c d} W^{a b c d}Z^e_f Z_e^f}{(D-3) (D-2)^5}\notag\\
   &+\frac{10
   (D-1)^3 D^2 \left((D-8) D^2+288\right) R \left(Z^a_b Z_a^b\right)^2}{3 (D-2)^6}-\frac{40 (D-1) D R^3 Z^a_b Z_a^b}{(D-2)^2}\notag\\
   &+\frac{5 (D-1)^2 D^2 \left(D^2-96\right) R^2 Z^a_b W_{a c d e}W^{bcde}}{2 (D-4) (D-3) (D-2)^2}+\frac{10 (D-1) D R^3 W_{a b c d} W^{a b c d}}{(D-3)
   (D-2)}\notag\\
   &+\frac{10 (D-1)^3 D^2 (D (11
   D-12)-284) R Z^{ab} W_{acbd} W^{c efg} W^d{}_{efg}}{3 (D-4) \left(D-3\right)^2 \left (D-2\right)^2}\, .
   \end{align}
   
\begin{align}   
   \mathcal{S}_{(5,3)}&=R^5+\frac{(D-1)^3 D^4 (D (5 D ((D-8) D+18)-16)-55) W_{a b c d} W^{a b c d}W\indices{_e_f^g^h}W\indices{_g_h^i^j}W\indices{_i_j^e^f}}{4 (D-3)^2 (D-2)^3 (D
   ((D-9) D+26)-22)}\notag\\
   &-\frac{5 (D-1)^3 D^4 (D (D (19 D-132)+155)+330) W_{a b c d} W^{a b c d}Z^e_f W_{e g h i}W^{fghi}}{8 (D-4)
   (D-3)^2 (D-2)^4}\notag\\
   &+\frac{15 (D-1)^2 D^3 (D (D ((D-12) D+30)+20)-55) R \left(W_{a b c d} W^{a b c d}\right)^2}{16
   (D-3)^2 (D-2)^4}\notag\\
   &+\frac{30 (D-1)^3 D^3 (D (D+9)-72) R Z^{a}_b Z_{a c} Z_{d e} W^{bdce}}{(D-4) (D-3)
   (D-2)^4}\notag\\
   &-\frac{5 (D-1)^3 D^3 (D (D (D (D (11 D-96)+97)+686)-768)+256) W_{a b c d} W^{a b c d}Z^e_f Z^f_gZ_e^g}{2
   (D-4) (D-3) (D-2)^6}\notag\\
   &+\frac{5  (D (D (D (D (8 D-71)+146)+147)-398)+184)
   W\indices{_a_b^c^d}W\indices{_c_d^e^f}W\indices{_e_f^a^b} Z^a_b Z_a^b}{2 (D-1)^{-3} D^{-3} (D-3) (D-2)^4 (D ((D-9) D+26)-22)}\notag\\
   &+\frac{10 (D-1)^4 D^3 (D ((D-8)
   D-30)+216) Z^a_b Z_a^bZ^c_d Z^e_f W\indices{_c_e^d^f}}{(D-3) (D-2)^6}\notag\\
   &+\frac{(D-1)^4 D^3 (D (5 D ((D-10)
   D+64)-512)-2048) Z^a_b Z_a^bZ^c_d Z^d_eZ_c^e}{2 (D-4) (D-2)^7}\notag\\
   &-\frac{5 (D-1)^2 D^2 (D ((D-3) D-13)+23) R^2
   W\indices{_a_b^c^d}W\indices{_c_d^e^f}W\indices{_e_f^a^b}}{2 (D-3) (D-2) (D ((D-9) D+26)-22)}\notag\\
   &-\frac{30 (D-9) (D-1)^2 D^2 R^2
   Z^a_b Z^c_d W\indices{_a_c^b^d}}{(D-3) (D-2)^3}-\frac{5 (D-1)^2 D^2 \left(D^2-32\right) R^2 Z^a_b Z^b_cZ_a^c}{(D-2)^4}\notag\\
   &+\frac{15
   (D-1)^2 D^2 \left(D^2-32\right) R^2 Z^a_b W_{a c d e}W^{bcde}}{2 (D-4) (D-3) (D-2)^2}\notag\\
   &+\frac{15 (D-1)^2 D^2
   (D (D (D ((D-17) D+23)+217)-320)+128) R W_{a b c d} W^{a b c d}Z^e_f Z_e^f}{4 (D-3) (D-2)^5}\notag\\
   &-\frac{5 (D-1)^3
   D^2 \left(D^2 ((D-13) D+64)-768\right) R \left(Z^a_b Z_a^b\right)^2}{4 (D-2)^6}\notag\\
   &-\frac{5 (D-1)^3 D^2 (D (D
   (2 D-25)+25)+184) R Z^{ab} W_{acbd} W^{c efg} W^d{}_{efg}}{(D-4) \left(D-2\right)^2\left(D-3\right)^2}\notag\\
   &+\frac{10 (D-1) D R^3 W_{a b c d} W^{a b c d}}{(D-3)
   (D-2)}-\frac{40 (D-1) D R^3 Z^a_b Z_a^b}{(D-2)^2}\, .
   \end{align}
   
\begin{align}   
   \mathcal{S}_{(5,4)}&=R^5-\frac{ (D (D (D (D (3 D-35)+150)-270)+151)+17) W_{a b c d} W^{a b c d}W\indices{_e_f^g^h}W\indices{_g_h^i^j}W\indices{_i_j^e^f}}{4 (D-1)^{-3} D^{-4}(D-3)^2
   (D-2)^3 (D ((D-9) D+26)-22)}\notag\\
   &+\frac{5 (D-1)^3 D^4 (D (D (9 D-64)+109)+66) W_{a b c d} W^{a b c d}
   Z^a_b W_{a c d e}W^{bcde}}{8 (D-3)^2 (D-2)^4}\notag\\
   &+\frac{15 (D-1)^2 D^3 (D (D ((D-8) D+18)-4)-11) R \left(W_{a b c d} W^{a b c d}\right)^2}{4
   (D-3)^2 (D-2)^4}\notag\\
   &+\frac{600 (D-1)^3 D^3 R Z^{a}_b Z_{a c} Z_{d e} W^{bdce}}{(D-3) (D-2)^4}-\frac{60 (D-5) (D-1)^2 D^2 R^2 Z^a_b Z^c_d W\indices{_a_c^b^d}}{(D-3) (D-2)^3}\notag\\
   &+\frac{5 (D-1)^3 D^3 (D
   (D (D (D (3 D-32)+75)+126)-176)+64) W_{a b c d} W^{a b c d}Z^e_f Z^f_gZ_e^g}{2 (D-3) (D-2)^6}\notag\\
   &-\frac{5 
   \left((D (D (3 D-41)+197)-343) D^3+344 D-176\right) W\indices{_a_b^c^d}W\indices{_c_d^e^f}W\indices{_e_f^a^b} Z^a_b Z_a^b}{2 (D-1)^{-3} D^{-3} (D-3) (D-2)^4 (D
   ((D-9) D+26)-22)}\notag\\
   &+\frac{10 (D-4) (D-1)^4 D^3 (D (D+3)-60) Z^a_b Z_a^bZ^c_d Z^e_f W\indices{_c_e^d^f}}{(D-3)
   (D-2)^6}\notag\\
   &-\frac{(D-1)^4 D^3 (D (D (D (3 D-26)+96)-256)-512) Z^a_b Z_a^bZ^c_d Z^d_eZ_c^e}{2 (D-2)^7}\notag\\
   &-\frac{5
   (D-1)^2 D^2 (D ((D-3) D-5)+11) R^2 W\indices{_a_b^c^d}W\indices{_c_d^e^f}W\indices{_e_f^a^b}}{(D-3) (D-2) (D ((D-9)
   D+26)-22)}\notag\\
   &-\frac{10 (D-4) (D-1)^2
   D^2 (D+4) R^2 Z^a_b Z^b_cZ_a^c}{(D-2)^4}+\frac{15 (D-1)^2 D^2 (D+4) R^2 Z^a_b W_{a c d e}W^{bcde}}{(D-3)
   (D-2)^2}\notag\\
   &+\frac{15 (D-1)^2 D^2 (D (D (D ((D-9) D+11)+53)-80)+32) R W_{a b c d} W^{a b c d}
   Z^a_b Z_a^b}{(D-3) (D-2)^5}\notag\\
   &-\frac{5 (D-4) (D-1)^3 D^2 (D ((D-5) D+12)+48) R
   \left(Z^a_b Z_a^b\right)^2}{(D-2)^6}\notag\\
   &-\frac{20 (D-1)^3 D^2 (2 (D-3) D-11) R Z^{ab} W_{acbd} W^{c efg} W^d{}_{efg}}{\left(D-2\right)^2\left(D-3\right)^2}\notag\\
   &+\frac{10
   (D-1) D R^3 W_{a b c d} W^{a b c d}}{(D-3) (D-2)}-\frac{40 (D-1) D R^3 Z^a_b Z_a^b}{(D-2)^2}\, .
   \end{align}
   
\begin{align}  
   \mathcal{S}_{(6,2)}&=R^6+\frac{15 (D-1) D W_{a b c d} W^{a b c d} R^4}{(D-3) (D-2)}-\frac{60 (D-1) D Z^a_b Z_a^b
   R^4}{(D-2)^2}\notag\\
   &-\frac{(D-1)^2 D^2 (D ((D-3) D-77)+119) W\indices{_a_b^c^d}W\indices{_c_d^e^f}W\indices{_e_f^a^b} R^3}{(D-3) (D-2) (D
   ((D-9) D+26)-22)}\notag\\
   &-\frac{12 (D-41) (D-1)^2 D^2 Z^a_b Z^c_d W\indices{_a_c^b^d} R^3}{(D-3) (D-2)^3}-\frac{2 (D-1)^2
   D^2 \left(D^2-160\right) Z^a_b Z^b_cZ_a^c R^3}{(D-2)^4}\notag\\
   &+\frac{3 (D-1)^2 D^2 \left(D^2-160\right) Z^a_b W_{a c d e}W^{bcde}
   R^3}{(D-4) (D-3) (D-2)^2}\notag\\
   &-\frac{9 (D-1)^2 D^3 (D ((D-3) D-27)+39) \left(W_{a b c d} W^{a b c d}\right)^2 R^2}{2
   (D-3)^2 (D-2)^4}\notag\\
   &+\frac{6 (D-1)^3 D^2 \left((D-8) D^2+480\right) \left(Z^a_b Z_a^b\right)^2 R^2}{(D-2)^6}\notag\\
   &+\frac{36
   (D-1)^3 D^3 (D (D+4)-164) Z^{a}_b Z_{a c} Z_{d e} W^{bdce} R^2}{(D-4) (D-3) (D-2)^4}\notag\\
   &-\frac{18 (D-1)^2 D^2
   (D (D (D (2 D-3)-139)+200)-80) W_{a b c d} W^{a b c d}Z^e_f Z_e^f R^2}{(D-3) (D-2)^5}\notag\\
   &+\frac{6 (D-1)^3 D^2
   (D (11 D-12)-476) Z^{ab} W_{acbd} W^{c efg} W^d{}_{efg} R^2}{(D-4) \left(D-2\right)^2 \left(D-3\right)^2}\notag\\
   &-\frac{6 (D-1)^3 D^4 (D
   ((D-3) D-13)+19) W_{a b c d} W^{a b c d}W\indices{_e_f^g^h}W\indices{_g_h^i^j}W\indices{_i_j^e^f} R}{(D-3)^2 (D-2)^3 (D ((D-9) D+26)-22)}\notag\\
   &-\frac{6 (D-1)^3
   D^3 (D (D (D (24 D-61)-623)+1096)-476) W\indices{_a_b^c^d}W\indices{_c_d^e^f}W\indices{_e_f^a^b} Z^g_h Z_g^h R}{(D-3) (D-2)^4 (D ((D-9)
   D+26)-22)}\notag\\
   &-\frac{24 (D-1)^4 D^3 (D (7 D+12)-492) Z^a_b Z_a^bZ^c_d Z^e_f W\indices{_c_e^d^f} R}{(D-3) (D-2)^6}\notag\\
   &+\frac{6
   (D-1)^3 D^3 (D (D (D (43 D-52)-1748)+1920)-640) W_{a b c d} W^{a b c d}Z^e_f Z^f_gZ_e^g R}{(D-4) (D-3)
   (D-2)^6}\notag\\
   &+\frac{192 (D-1)^4 D^3 ((D-8) D-32) Z^a_b Z_a^bZ^c_d Z^d_eZ_c^e R}{(D-4) (D-2)^7}\notag\\
   &+\frac{3 (D-1)^3
   D^4 (D (23 D-32)-468) W_{a b c d} W^{a b c d}Z^e_f W_{e g h i}W^{fghi} R}{(D-4) (D-3)^2 (D-2)^4}\notag\\
   &-\frac{5 (D-1)^3 D^5
   (D ((D-3) D-7)+11) \left(W_{a b c d} W^{a b c d}\right)^3}{2 (D-3)^3 (D-2)^6}\notag\\
   &-\frac{16 (D-1)^5 D^3 (D (D
   (D+8)-80)-160) \left(Z^a_b Z_a^b\right)^3}{(D-2)^9}\notag\\
   &-\frac{6 (D-1)^4 D^3 (D (D (D (47 D-48)-1272)+1120)-320)
   W_{a b c d} W^{a b c d}\left(Z^e_f Z_e^f\right)^2}{(D-3) (D-2)^8}\notag\\
   &-\frac{6 (D-1)^3 D^4 (D (D (D (20 D-59)-245)+508)-234)
   \left(W_{a b c d} W^{a b c d}\right)^2Z^e_f Z_e^f}{(D-3)^2 (D-2)^7}\notag\\
   &+\frac{384 (D-1)^5 D^4 \left(D^2+D-41\right) Z^{a}_b Z_{a c} Z_{d e} W^{bdce}
   Z^f_g Z_f^g}{(D-4) (D-3) (D-2)^7}\notag\\
   &+\frac{2 (D (D (3 D (61
   D-132)-3368)+5008)-1904) W\indices{_a_b^c^d}W\indices{_c_d^e^f}W\indices{_e_f^a^b}Z^g_h Z^h_iZ_g^i}{(D-1)^{-4} D^{-4}(D-4) (D-3) (D-2)^5 (D ((D-9) D+26)-22)}\notag\\
   &+\frac{6
   (D-1)^4 D^4 (D (13 D-20)-152) W_{a b c d} W^{a b c d}Z^{ef} W_{egfh} W^{g ijk} W^h{}_{ijk}}{(D-4) (D-3)^3 (D-2)^4}\, .
   \end{align}
\begin{align}   
   \mathcal{S}_{(6,3)}&=R^6+\frac{15 (D-1) D W_{a b c d} W^{a b c d} R^4}{(D-3) (D-2)}-\frac{60 (D-1) D Z^a_b Z_a^b R^4}{(D-2)^2}\notag\\
   &-\frac{(D-1)^2 D^2
   (D (3 (D-3) D-71)+117) W\indices{_a_b^c^d}W\indices{_c_d^e^f}W\indices{_e_f^a^b} R^3}{(D-3) (D-2) (D ((D-9) D+26)-22)}\notag\\
   &-\frac{12 (D-1)^2 D^2 (3
   D-43) Z^a_b Z^c_d W\indices{_a_c^b^d} R^3}{(D-3) (D-2)^3}-\frac{2 (D-1)^2 D^2 \left(3 D^2-160\right) Z^a_b Z^b_cZ_a^c R^3}{(D-2)^4}\notag\\
   &+\frac{3
   (D-1)^2 D^2 \left(3 D^2-160\right) Z^a_b W_{a c d e}W^{bcde} R^3}{(D-4) (D-3) (D-2)^2}\notag\\
   &+\frac{9 (D-1)^2 D^3 (D (D
   ((D-16) D+42)+80)-147) \left(W_{a b c d} W^{a b c d}\right)^2 R^2}{8 (D-3)^2 (D-2)^4}\notag\\
   &-\frac{3 (D-1)^3 D^2 \left(D^2 ((D-17)
   D+96)-1920\right) \left(Z^a_b Z_a^b\right)^2 R^2}{2 (D-2)^6}\notag\\
   &+\frac{36 (D-1)^3 D^3 (D (2 D+13)-172) Z^{a}_b Z_{a c} Z_{d e} W^{bdce} R^2}{(D-4) (D-3)
   (D-2)^4}\notag\\
   &+\frac{9 (D-1)^2 D^2 (D (D (D ((D-25) D+35)+549)-800)+320) W_{a b c d} W^{a b c d}Z^e_f Z_e^f R^2}{2 (D-3)
   (D-2)^5}\notag\\
   &-\frac{6 (D-1)^3 D^2 (D (2 (D-18) D+37)+468) Z^{ab} W_{acbd} W^{c efg} W^d{}_{efg} R^2}{(D-4) \left(D-2\right)^2 \left(D-3\right)^2}\notag\\
   &+\frac{3 (D-1)^3 D^4 (D ((D-6) (D-4) D+10)-33) W_{a b c d} W^{a b c d}W\indices{_e_f^g^h}W\indices{_g_h^i^j}W\indices{_i_j^e^f} R}{(D-3)^2 (D-2)^3
   (D ((D-9) D+26)-22)}\notag\\
   &+\frac{6  (D (D ((D-9) D (8 D-23)+498)-1046)+468) W\indices{_a_b^c^d}W\indices{_c_d^e^f}W\indices{_e_f^a^b}
   Z^g_h Z_g^h R}{(D-1)^{-3} D^{-3}(D-3) (D-2)^4 (D ((D-9) D+26)-22)}\notag\\
   &+\frac{24 (D-1)^4 D^3 (D ((D-15) D-42)+516)
   Z^a_b Z_a^bZ^c_d Z^e_f W\indices{_c_e^d^f} R}{(D-3) (D-2)^6}\notag\\
   &-\frac{6  (D (D (D (D (11 D-139)+149)+1730)-1920)+640)
   W_{a b c d} W^{a b c d}Z^e_f Z^f_gZ_e^g R}{(D-1)^{-3} D^{-3}(D-4) (D-3) (D-2)^6}\notag\\
   &+\frac{6 (D-1)^4 D^3 (D (D ((D-10) D+96)-256)-1024) Z^a_b Z_a^b
   Z^c_d Z^d_eZ_c^e R}{(D-4) (D-2)^7}\notag\\
   &-\frac{3 (D-1)^3 D^4 (D (D (19 D-178)+219)+882) W_{a b c d} W^{a b c d}Z^e_f W_{e g h i}W^{fghi} R}{2 (D-4)
   (D-3)^2 (D-2)^4}\notag\\
   &+\frac{5 (D-1)^3 D^5 (D (3 D ((D-8) D+18)-8)-33) \left(W_{a b c d} W^{a b c d}\right)^3}{8 (D-3)^3
   (D-2)^6}\notag\\
   &+\frac{4 (D-1)^5 D^3 (D (3 (D-32) D+320)+640) \left(Z^a_b Z_a^b\right)^3}{(D-2)^9}\notag\\
   &+\frac{3  (D
   (D (D (D (73 D-625)+544)+5048)-4480)+1280) W_{a b c d} W^{a b c d}\left(Z^e_f Z_e^f\right)^2}{2 (D-1)^{-4} D^{-3} (D-3) (D-2)^8}\notag\\
   &+\frac{3  (D
   (D (3 D (D (15 D-133)+309)+313)-1728)+882) \left(W_{a b c d} W^{a b c d}\right)^2Z^e_f Z_e^f}{2 (D-1)^{-3} D^{-4} (D-3)^2 (D-2)^7}\notag\\
   &-\frac{3 (D-1)^4 D^4 (D (D (17 D-118)+151)+264) W_{a b c d} W^{a b c d}Z^{ef} W_{egfh} W^{g ijk} W^h{}_{ijk}}{(D-4) (D-3)^3 (D-2)^4}\notag\\
   &-\frac{2(D (D (3
   D (D (29 D-254)+458)+2672)-4772)+1872) W\indices{_a_b^c^d}W\indices{_c_d^e^f}W\indices{_e_f^a^b}Z^g_h Z^h_iZ_g^i}{(D-1)^{-4} D^{-4}(D-4) (D-3) (D-2)^5 (D ((D-9)
   D+26)-22)}\notag\\
   &-\frac{12 (D-1)^5 D^4 (3
   D (3 (D-8) D-40)+1376) Z^{a}_b Z_{a c} Z_{d e} W^{bdce} Z^f_g Z_f^g}{(D-4) (D-3) (D-2)^7}\, .
\end{align}

\begin{small}
\begin{align}   
   \mathcal{S}_{(6,4)}&=R^6+\frac{15 (D-1) D W_{a b c d} W^{a b c d} R^4}{(D-3) (D-2)}-\frac{60 (D-1) D Z^a_b Z_a^b R^4}{(D-2)^2}-\frac{24 (3 D-23) Z^a_b Z^c_d W\indices{_a_c^b^d} R^3}{D^{-2} (D-1)^{-2} (D-3) (D-2)^3}\notag\\
   &-\frac{2
   (D-1)^2 D^2 (D (3 (D-3) D-31)+57) W\indices{_a_b^c^d}W\indices{_c_d^e^f}W\indices{_e_f^a^b} R^3}{(D-3) (D-2) (D ((D-9)
   D+26)-22)}\notag\\
   &-\frac{4 (D-1)^2 D^2
   \left(3 D^2-80\right) Z^a_b Z^b_cZ_a^c R^3}{(D-2)^4}+\frac{6 (D-1)^2 D^2 \left(3 D^2-80\right) Z^a_b W_{a c d e}W^{bcde} R^3}{(D-4)
   (D-3) (D-2)^2}\notag\\
   &+\frac{9 (D-1)^2 D^3 (D ((D-6) (D-4) D+8)-33) \left(W_{a b c d} W^{a b c d}\right)^2 R^2}{2 (D-3)^2
   (D-2)^4}\notag\\
   &-\frac{6 (D-1)^3 D^2 \left(D^2 ((D-11) D+48)-480\right) \left(Z^a_b Z_a^b\right)^2 R^2}{(D-2)^6}\notag\\
   &+\frac{72
   (D-1)^3 D^3 (D (D+14)-92) Z^{a}_b Z_{a c} Z_{d e} W^{bdce} R^2}{(D-4) (D-3) (D-2)^4}\notag\\
   &+\frac{18 (D-1)^2 D^2
   (D (D (D ((D-13) D+17)+135)-200)+80) W_{a b c d} W^{a b c d}Z^e_f Z_e^f R^2}{(D-3) (D-2)^5}\notag\\
   &-\frac{12 (D-1)^3
   D^2 (D (D (4 D-39)+38)+228) Z^{ab} W_{acbd} W^{c efg} W^d{}_{efg} R^2}{(D-4) \left(D-2\right)^2 \left(D-3\right)^2}\notag\\
   &-\frac{3(D (D (D (D (3 D-55)+310)-630)+215)+237) W_{a b c d} W^{a b c d}W\indices{_e_f^g^h}W\indices{_g_h^i^j}W\indices{_i_j^e^f} R}{10(D-1)^{-3} D^{-4}(D-3)^2 (D-2)^3 (D
   ((D-9) D+26)-22)}\notag\\
   &-\frac{3(D (D (D (D (D (3
   D-73)+481)-927)-588)+1936)-912) W\indices{_a_b^c^d}W\indices{_c_d^e^f}W\indices{_e_f^a^b} Z^g_h Z_g^h R}{(D-1)^{-3} D^{-3}(D-3) (D-2)^4 (D ((D-9) D+26)-22)}\notag\\
   &+\frac{12 (D-1)^4
   D^3 (D (D (5 D-33)-192)+1104) Z^a_b Z_a^bZ^c_d Z^e_f W\indices{_c_e^d^f} R}{(D-3) (D-2)^6}\notag\\
   &+\frac{3
   (D (D (D (D (D (3 D-88)+587)-562)-3424)+3840)-1280) W_{a b c d} W^{a b c d}Z^e_f Z^f_gZ_e^g R}{(D-1)^{-3} D^{-3}(D-4) (D-3)
   (D-2)^6}\notag\\
   &-\frac{3 (D-1)^4 D^3 (D (D (D (D (3 D-58)+400)-1920)+2560)+10240) Z^a_b Z_a^bZ^c_d Z^d_eZ_c^e R}{5
   (D-4) (D-2)^7}\notag\\
   &+\frac{3 (D-1)^3 D^4 (D (D (D (9 D-176)+893)-990)-1584) W_{a b c d} W^{a b c d}Z^e_f W_{e g h i}W^{fghi} R}{4
   (D-4) (D-3)^2 (D-2)^4}\notag\\
   &-\frac{(D-1)^3 D^5 (D (3 D (D (D (3 D-35)+150)-270)+445)+51)
   \left(W_{a b c d} W^{a b c d}\right)^3}{8 (D-3)^3 (D-2)^6}\notag\\
   &+\frac{(D-1)^5 D^3 (D (3 D (3 D ((D-12)
   D+80)-1280)+6400)+12800) \left(Z^a_b Z_a^b\right)^3}{5 (D-2)^9}\notag\\
   &-\frac{3 (D (D (D (D (D (21
   D-298)+1333)-1000)-5024)+4480)-1280) W_{a b c d} W^{a b c d}\left(Z^e_f Z_e^f\right)^2}{2(D-1)^{-4} D^{-3}(D-3) (D-2)^8}\notag\\
   &-\frac{3(D (D (3
   D (D (D (13 D-172)+826)-1592)+1891)+2388)-1584) \left(W_{a b c d} W^{a b c d}\right)^2Z^e_f Z_e^f}{4 (D-1)^{-3} D^{-4}(D-3)^2 (D-2)^7}\notag\\
   &+\frac{6 (D-1)^5
   D^4 (3 D (D ((D-16) D+56)+192)-2944) Z^{a}_b Z_{a c} Z_{d e} W^{bdce} Z^f_g Z_f^g}{(D-4) (D-3)
   (D-2)^7}\notag\\
   &+\frac{(D (D (3 D (D (D (21 D-286)+1345)-2092)-3040)+8800)-3648)
   W\indices{_a_b^c^d}W\indices{_c_d^e^f}W\indices{_e_f^a^b}Z^g_h Z^h_iZ_g^i}{(D-1)^{-4} D^{-4}(D-4) (D-3) (D-2)^5 (D ((D-9) D+26)-22)}\notag\\
   &+\frac{3(D
   (D (D (87 D-950)+3463)-3784)-1896) W_{a b c d} W^{a b c d}Z^{ef} W_{egfh} W^{g ijk} W^h{}_{ijk}}{10(D-1)^{-4} D^{-4}(D-4) (D-3)^3 (D-2)^4}\, .
   \end{align}
\end{small}

\begin{small}
\begin{align}   
   \mathcal{S}_{(6,5)}&= R^6+\frac{15 (D-1) D W_{a b c d} W^{a b c d} R^4}{(D-3) (D-2)}-\frac{60 (D-1) D Z^a_b Z_a^b R^4}{(D-2)^2}\notag\\
   &-\frac{10
   (D-1)^2 D^2 (D ((D-3) D-5)+11) W\indices{_a_b^c^d}W\indices{_c_d^e^f}W\indices{_e_f^a^b} R^3}{(D-3) (D-2) (D ((D-9)
   D+26)-22)}\notag\\
   &-\frac{120 (D-5) (D-1)^2 D^2 Z^a_b Z^c_d W\indices{_a_c^b^d} R^3}{(D-3) (D-2)^3}-\frac{20 (D-4) (D-1)^2
   D^2 (D+4) Z^a_b Z^b_cZ_a^c R^3}{(D-2)^4}\notag\\
   &+\frac{30 (D-1)^2 D^2 (D+4) Z^a_b W_{a c d e}W^{bcde} R^3}{(D-3)
   (D-2)^2}\notag\\
   &+\frac{45 (D-1)^2 D^3 (D (D ((D-8) D+18)-4)-11) \left(W_{a b c d} W^{a b c d}\right)^2 R^2}{4 (D-3)^2
   (D-2)^4}\notag\\
   &-\frac{15 (D-4) (D-1)^3 D^2 (D ((D-5) D+12)+48) \left(Z^a_b Z_a^b\right)^2 R^2}{(D-2)^6}\notag\\
   &+\frac{1800
   (D-1)^3 D^3 Z^{a}_b Z_{a c} Z_{d e} W^{bdce} R^2}{(D-3) (D-2)^4}\notag\\
   &+\frac{45 (D-1)^2 D^2 (D (D (D ((D-9)
   D+11)+53)-80)+32) W_{a b c d} W^{a b c d}Z^e_f Z_e^f R^2}{(D-3) (D-2)^5}\notag\\
   &-\frac{60 (D-1)^3 D^2 (2 (D-3) D-11) Z^{ab} W_{acbd} W^{c efg} W^d{}_{efg}
   R^2}{\left(D-2\right)^2 \left(D-3\right)^2}\notag\\
   &-\frac{3  (D (D (D (D (3 D-35)+150)-270)+151)+17)
   W_{a b c d} W^{a b c d}W\indices{_e_f^g^h}W\indices{_g_h^i^j}W\indices{_i_j^e^f} R}{2 (D-1)^{-3} D^{-4} (D-3)^2 (D-2)^3 (D ((D-9) D+26)-22)}\notag\\
   &-\frac{15 (D-1)^3 D^3 \left((D
   (D (3 D-41)+197)-343) D^3+344 D-176\right) W\indices{_a_b^c^d}W\indices{_c_d^e^f}W\indices{_e_f^a^b} Z^g_h Z_g^h R}{(D-3) (D-2)^4 (D ((D-9)
   D+26)-22)}\notag\\
   &+\frac{60 (D-4) (D-1)^4 D^3 (D (D+3)-60) Z^a_b Z_a^bZ^c_d Z^e_f W\indices{_c_e^d^f} R}{(D-3) (D-2)^6}\notag\\
   &+\frac{15
   (D-1)^3 D^3 (D (D (D (D (3 D-32)+75)+126)-176)+64) W_{a b c d} W^{a b c d}Z^e_f Z^f_gZ_e^g R}{(D-3)
   (D-2)^6}\notag\\
   &-\frac{3 (D-1)^4 D^3 (D (D (D (3 D-26)+96)-256)-512) Z^a_b Z_a^bZ^c_d Z^d_eZ_c^e
   R}{(D-2)^7}\notag\\
   &+\frac{15 (D-1)^3 D^4 (D (D (9 D-64)+109)+66) W_{a b c d} W^{a b c d}Z^e_f W_{e g h i}W^{fghi} R}{4 (D-3)^2
   (D-2)^4}\notag\\
   &+\frac{5 (D-1)^3 D^5 (D (D (D (D ((D-15) D+90)-270)+405)-251)+32) \left(W_{a b c d} W^{a b c d}\right)^3}{8
   (D-3)^3 (D-2)^6}\notag\\
   &-\frac{(D-4) (D-1)^5 D^3 (D (D (D ((D-12) D+60)-200)+480)+640)
   \left(Z^a_b Z_a^b\right)^3}{(D-2)^9}\notag\\
   &+\frac{15 (D-4)(D (D (D (D ((D-16)
   D+83)-124)-200)+208)-64) W_{a b c d} W^{a b c d}\left(Z^e_f Z_e^f\right)^2}{2 (D-1)^{-4} D^{-3}  (D-3) (D-2)^8}\notag\\
   &+\frac{15  (D (D (D
   (D (D (D (3 D-52)+358)-1220)+2019)-1208)-148)+264) \left(W_{a b c d} W^{a b c d}\right)^2Z^e_f Z_e^f}{4(D-1)^{-3} D^{-4} (D-3)^2 (D-2)^7}\notag\\
   &-\frac{3 (D-1)^4 D^4 (D (D (D (8 D-85)+312)-401)-34) W_{a b c d} W^{a b c d}Z^{ef} W_{egfh} W^{g ijk} W^h{}_{ijk}}{2(D-3)^3 (D-2)^4}
   \notag\\
   &-\frac{5 
   (D (D (D (D (D (4 D-57)+288)-565)+102)+344)-176) W\indices{_a_b^c^d}W\indices{_c_d^e^f}W\indices{_e_f^a^b}Z^g_h Z^h_iZ_g^i}{(D-1)^{-4} D^{-4}(D-3) (D-2)^5 (D
   ((D-9) D+26)-22)}\notag\\
   &+\frac{30
    ((D-8) D (D+2)+160) Z^{a}_b Z_{a c} Z_{d e} W^{bdce} Z^f_g Z_f^g}{(D-1)^{-5} D^{-4}(D-3) (D-2)^7}\, .
   \end{align} 
\end{small}
   
\bibliographystyle{JHEP-2}
\bibliography{Gravities.bib}

\end{document}